\documentclass[11pt]{article}

\usepackage{amsfonts,amssymb,amsmath,latexsym,enumitem,ae,aecompl,color}

\newcommand{\FF}{\vspace*{\medskipamount}}

\newcommand{\BBB}{\vspace*{-\bigskipamount}}

\newcommand{\cE}{\mathcal{E}}

\newcommand{\cO}{\mathcal{O}}

\newcommand{\cS}{\mathcal{S}}

\newcommand{\Paragraph}[1]{\BBB\paragraph{#1}}

\newlength{\pagewidth}
\newlength{\captionwidth}
\newcommand{\qed}{\hfill $\square$ \smallbreak}
\newenvironment{proof}{\noindent{\bf Proof:}}{\qed}

\newtheorem{theorem}{Theorem}
\newtheorem{lemma}{Lemma}

\newtheorem{corollary}{Corollary}

\begin{document}

\baselineskip          	3ex
\parskip                	1ex


\title{				Asynchronous Exclusive Selection \footnotemark[1]\vfill\vfill}

\author{			Bogdan S. Chlebus  	\footnotemark[2]   	
				\and
				Dariusz R. Kowalski 	\footnotemark[3]}

\footnotetext[1]{	A preliminary version of this paper appeared as~\cite{ChlebusK08}.}

\footnotetext[2]{		School of Computer and Cyber Sciences,
               			Augusta University,
               			Augusta, Georgia, USA.
Work supported by the National Science Foundation Grant No.~1016847.}

\footnotetext[3]{		School of Computer and Cyber Sciences,
               			Augusta University,
               			Augusta, Georgia, USA.
Work supported by the National Science Foundation Grant No.~2131538 and  the Polish National Science Center NCN grant UMO-2017/25/B/ST6/02553.}

\date{}

\maketitle

\vfill

\begin{abstract}

We consider the task of assigning unique integers to a group of processes in an asynchronous distributed system of a total of $n$ processes prone to crashes that communicate through shared read-write registers.
In the Renaming problem, an arbitrary group of $k\le n$ processes that hold the original names from a range $[N]=\{1,\ldots,N\}$, contend to acquire unique integers in a smaller range~$[M]$  as new names using some~$r$ auxiliary shared registers.
We develop a wait-free $(k,N)$-renaming solution, where both $k$ and $N$ are known, operating in $\cO(\log k (\log N + \log k\log^* \frac{N}{k}))$ local steps, for $M=\cO(k)$, and with $r=\cO(k\log\frac{N}{k})$ auxiliary registers.
We give a wait-free  $N$-renaming algorithm, where $N$ is known,  operating in $\cO(\log^2 k \,(\log N + \log k\log^\ast N))$ local steps, with $M=\cO(k)$ and with $r=\cO(n\log\frac{N}{n})$  registers.
We develop a wait-free $k$-renaming algorithm, where $k$ known, operating in $\cO(k)$ time, for $M=2k-1$ and with $r=\cO(k^2)$ registers.
We give an adaptive wait-free solution of Renaming, where neither $k$ nor $N$ is known, having  $M=8k-\lg k-1$ as a bound on the range of new names, which operates in~$\cO(k)$ local steps and uses $r=\cO(n^2)$  registers.
As a  lower bound, we show that a wait-free solution to Renaming requires $1+\min\{k-2,\lfloor\log_{2r} \frac{N}{M+k-1}\rfloor\}$ steps in the worst case.
We apply  renaming algorithms to obtain solutions to Store\&Collect problem, which is about a group of $k\le n$ processes with the original names in a range $[N]$ proposing individual values (operation \texttt{Store}) and returning a view of all proposed values (operation \texttt{Collect}), while using some $r$ auxiliary shared read-write registers.
We show that if a known $N$ is polynomial in~$n$, then storing can be performed in $\cO(\log^3 n \log^\ast n)$ local steps and collecting in $\cO(k)$ local steps with $\cO(n\log n)$ shared read-write registers. 
We consider a problem Mining-Names, in which processes may repeatedly request  positive integers as new names subject to the constraints that no integer can be assigned to different processes and the number of integers never acquired as names is finite in an infinite execution.
We give two solutions to Mining-Names in a distributed system in which there are infinitely many shared read-write registers available.
A non-blocking solution leaves at most $2n-2$ nonnegative integers never assigned as names, and a wait-free algorithm leaves at most $(n+2)(n-1)$ nonnegative integers never assigned as names.

\vfill

\end{abstract}

\vfill

~

\thispagestyle{empty}

\setcounter{page}{0}
\newpage

\section{Introduction}

We consider asynchronous distributed systems consisting of some $n$ processes that are prone to crashes and use read-write registers for inter-process communication.
The studied problems concern assigning positive integers to the processes in an \emph{exclusive} fashion, which means that no integer is assigned to two distinct processes.
We seek wait-free algorithms, and sometimes consider non-blocking ones.

When an integer $i$ is assigned to a process~$p$ exclusively, then we say that $i$ is $p$'s  \emph{new name}. 
In the Renaming problem, some $k\le n$ processes, each having an original name from a large range $[N]=\{1,\ldots,N\}$, contend to acquire unique integers in a smaller range~$[M]$  as new names, using some $r$ shared registers.
An algorithm can have some of the parameters $k$ and $N$ as a part of code.
We indicate which parameters are known as a part of code by attaching the relevant parameters to problems' names and solutions.
For example, an algorithm solving $(k,N)$-renaming has both~$k$ and~$N$ as a part of code, while $M$ and $r$ and the time complexity are characteristics of the algorithm, given as functions of $k$ and $N$ and possibly also of~$n$.
Similarly, an algorithm solving $k$-renaming works for any range $[N]$ of the original names and for up to $k$ contending processes.
Finally, an $N$-renaming algorithm handles the original names in the range $[N]$ while the contention $k$ is arbitrary, except for the restriction~$k\le n$.
An adaptive renaming algorithm works for any contention~$k\le n$ and range $[N]$ of the original names, which are not parts of code, while the range of new names~$M$ and the time performance are  functions of~$k$, and the number of registers~$r$ is a function of~$n$, as this is the maximum value of $k$ when $k\le n$.

We restrict our attention to \emph{one-time} renaming problems in which processes that will contend to acquire new names are designated at the start of an execution, and no new name is ever released and reused. 
Time performance is measure by the number of \emph{local steps}, which is a maximum number of steps a process takes before halting in a final state.

In the problem Store\&Collect, some $k$ processes perform two operations \texttt{Store} and \texttt{Collect}.
The \texttt{Store} operation by a process~$p$ proposes a value, and \texttt{Collect} results in returning a \emph{view}, which is a collection of pairs $(p,v)$ where $p$ is the original name of a process that proposed a value~$v$ before the return of \texttt{Collect}, but not a stale one replaced before the invocation of \texttt{Collect}.
The semantics of this problem under asynchrony is well determined by referring to a collection of read-write registers, one register assigned to each process.
In order to propose a value, a process writes its original name and that value in its register.
In order to collect, a process reads once each register storing a pair consisting of a value proposed by a process and its name, in arbitrary order of registers, and includes each such a read pair in the view.

The problem Mining-Names is about a scenario in which processes repeatedly request new positive integers as names in an infinite execution.
There are two constraints on algorithms for the problem. 
One is that each process keeps an exclusive ownership of each acquired new name, to possibly build an infinite collection of new names in an infinite execution.
The other is to leave only finitely many positive integers never assigned as new names in an infinite execution. 
For a fixed integer~$i$, no wait-free solution to Mining-Names can guarantee that $i$ is eventually assigned as a new name.
It follows that some integers may never be used as new names when an algorithm works to assign as many positive integers as new names as possible.
We want to minimize the number of positive integers never assigned as names in an execution, so this number  is proposed as a measure of quality of a solution for Mining-Names.
The model of a distributed system we use to develop solutions for Mining-Names assumes finitely many processes but infinitely many shared read-write registers.

As a preparation to developing Mining-Names solutions, we show how  to implement a repository of infinitely many values using infinitely many  read-write registers. 
The values are generated in a dynamic fashion.
A value is considered as \emph{deposited} in a register when the value is stored in the  register and will never be overwritten.
In this problem, we strive to minimize the number of available registers  that never become used to store deposited values.
The problem is closely related to mining names, as new names can be used to identify registers to make deposits.


\Paragraph{Contributions of this paper.}

The renaming algorithms that we develop are designed to have processes traverse paths in graphs in which vertices represent names.
During such concurrent traversals, processes compete to acquire the names of the visited vertices.
We consider graphs with suitable expansion properties as means to makes algorithms efficient.
The approaches to renaming known in the literature, that have graphs built into algorithms in a similar manner, use graphs with  simple regular topologies, like constant-degree grids.

We develop a wait-free $(k,N)$-renaming algorithm operating in $\cO(\log k (\log N + \log k\log^* \frac{N}{k}))$ local steps, with a range of new names $M=\cO(k)$ and $r=\cO(k\log\frac{N}{k})$ auxiliary registers.
This is a first known deterministic algorithm with step complexity polylogarithmic in~$k$ and a range of new names $M$ linear in~$k$, for~$N$ that is polynomial in~$k$.

We show that  $1+\min\{k-2,\lfloor\log_{2r} \frac{N}{M+k-1}\rfloor\}$ local steps are required in the worst case by any wait-free $(k,N)$-renaming algorithm to assign names from a  range~$[M]$ when using~$r$ registers.
This is a first known lower bound on the local-step time performance of Renaming  that comprehensively involves all the four parameters $k$, $N$, $M$, and~$r$.
In particular, if $N$ is unknown, and hence could be arbitrarily large, while $M$ is suitably bounded as a function of~$k$, then $k-1$ is a lower bound; this resembles the lower bounds given by Jayanti et al.~\cite{JayantiTT00}.
An $\Omega(k)$ lower bound on time, valid under additional assumptions, was proven by Alistarh et al.~\cite{AlistarhACGG14}   by a different argument.

We develop a wait-free $N$-renaming algorithm with $\cO(\log^2 k \,(\log N + \log k\log^\ast N))$ local step complexity, for unknown contention $k$, with the range of new names 
$M=\cO(k)$ and the number of registers $r=\cO(n\log\frac{N}{n})$. 
If a known~$N$ is poly-logarithmic in~$n$, then this renaming algorithm runs in $\cO(\log^3 n \log^\ast n)$ local steps and uses $\cO(n\log n)$ auxiliary registers.

We develop a wait-free $k$-renaming algorithm operating in $\cO(k)$ local steps, with a bound on new names $M=2k-1$ and with $r=\cO(k^2)$ auxiliary registers.
The time complexity of this algorithm is asymptotically optimal, which follows from the lower bound we show and the property that the algorithm works for arbitrary~$N$.
This is the first algorithm known that has simultaneously two properties: the local step complexity is $\cO(k)$ and also the range of new names is~$M=\cO(k)$.
Among the previously known algorithms that run in time $\cO(k)$, the value $M=\frac{k(k+1)}{2}$ was smallest known; it is achieved by an algorithm of Moir and Anderson~\cite{MoirA95}.
The fastest algorithm known before, among those having $M=\cO(k)$, operates in $\cO(k\log k)$ time, it was given by Attiya and Fouren~\cite{AttiyaF01}.
The value $M=2k-1$ is known to be the best possible size of a range of names for infinitely many values of~$k$, as showed by Herlihy and Shavit~\cite{HerlihyS99}. 
The fastest algorithm known prior to this work with $M=2k-1$ as a bound on the range of new names runs in time $\cO(k^2)$, it was given by Afek and Merritt~\cite{AfekM99}.

We give a fully adaptive  renaming algorithm, with neither $k$ nor $N$ known, having a bound on the range of new names~$M$ as small as $8k-\lg k-1$.
The algorithm operates in~$\cO(k)$ local steps and uses $\cO(n^2)$ auxiliary registers.
This is an improvement with respect to time performance over the previously known algorithms.
The algorithm of Moir and Anderson~\cite{MoirA95} works in time $\cO(k)$ with a range of new names $M=\cO(k^2)$ and using $\cO(n^2)$ registers.
The algorithm of Attiya and Fouren~\cite{AttiyaF01} operates in $\cO(k\log k)$ time with a range $M=\cO(k)$.
The algorithm of Afek and Merritt~\cite{AfekM99} runs in time $\cO(k^2)$ with a range of new names~$M=2k-1$.

We apply  renaming algorithms to obtain solutions to Store\&Collect of the following performing characteristics.
When both parameters $k$ and $N$ are known, then Store\&Collect can be implemented such that the first storing operation by a process takes  $\cO(\log k (\log N + \log k\log^* \frac{N}{k}))$ steps and collecting $\cO(k)$ steps, while using $\cO(k\log\frac{N}{k})$ auxiliary registers.
If $N$ is known but $k$ is not, then Store\&Collect can be implemented such that the first instance of storing takes $\cO(\log^2 k (\log N + \log k\log^\ast N))$ local steps and collecting $\cO(k)$ steps, with $\cO(n\log\frac{N}{n})$ auxiliary registers.
If the number of participating processes~$k$ is known but the range of original names~$N$ is not, then Store\&Collect can be implemented such that the first instance of storing takes $\cO(k)$ local steps and collecting $\cO(k)$ steps, with $\cO(k^2)$ auxiliary registers.
If $k$ and $N$ are both unknown, which is the adaptive case, then Store\&Collect can be implemented such that storing takes $\cO(k)$ steps and collecting takes $\cO(k)$ steps, with $\cO(n^2)$ auxiliary registers.
Afek and De Levie~\cite{AfekL07} gave an adaptive solution to Store\&Collect achieving storing in $\cO(k)$ local steps and collecting in $\cO(k)$ local steps, for~$r=\cO(n^2)$.

We consider the problem called Mining-Names, which is about processes working to claim nonnegative integers  as names in a mutually exclusive manner.
We show that Mining-Names is solvable in a non-blocking way such that at most $2n-2$ integers are never assigned as names, which is asymptotically best possible, and in a wait-free manner so that at most $(n+2)(n-1)$ values are never assigned as names.
The problem Mining-Names  has not been considered  before in the literature, by the knowledge of the authors of the paper.
Name mining  is related to ``depositing'' infinitely many values in read-write registers, where depositing means storing a value in a register such that it is never overwritten.
We give a non-blocking implementation of depositing in which at most $2n-2$ dedicated read-write registers are never used for depositing, and a wait-free implementation with the property that at most $(n+2)(n-1)$ dedicated deposit registers are never used for depositing.


\Paragraph{Related work.}

Now we describe the context of this work by reviewing research related to renaming.
We restrict our attention to asynchronous systems with shared memory consisting of read-write registers only; for a comprehensive survey of renaming see Alistarh~\cite{Alistarh15}.

The problem of renaming was introduced by Attiya et al.~\cite{AttiyaBDPR90} in the model of asynchronous message-passing. 
They showed that $n$ processes may assign themselves new names from the range $[n+f]$, where $f<n$ is an upper bound on the number of crashes.
This established renaming as a non-trivial algorithmic problem with a wait-free solution for  environments in which Consensus cannot be solved; see~\cite{Attiya-Welch-book2004,HerlihyKozlovRajsbaum-book,Lynch-book96}.
The range of new names $[M=n+f]$, with up to $f$ crashes, was shown to be the smallest possible for renaming to be solvable by Herlihy and Shavit~\cite{HerlihyS99}.
Next we consider a scenario win which $k\leq n$ contending processes with original names in a large range~$[N]$ want to obtain new names in a small range~$[M]$.
Borowsky and Gafni~\cite{BorowskyG92} gave a wait-free algorithm solving this problem in  time $\cO(k^2N)$ for~$M=2k-1$.
Moir and Anderson~\cite{MoirA95} gave a solution with time complexity $\cO(k)$, for new names of magnitude $M=\frac{k(k+1)}{2}$ and using $r=\cO(k^2)$ registers. 
Attiya and Fouren~\cite{AttiyaF01} gave an algorithm working in time $\cO(k\log k)$ for $M=6k-1$, and another of time complexity $\cO(N)$ for $M=2k-1$. 
Afek and Merritt~\cite{AfekM99} developed an algorithm working in time $\cO(k^2)$ for $M=2k-1$.

A  renaming solution is \emph{long-lived} when processes may invoke the operations to request a name and to release the current name repeatedly, as long as exclusiveness of each name holds within the interval from acquiring to releasing.
It is assumed that at most $k$ processes contend for names concurrently.
The following is a selection of published long-lived renaming algorithms.
Burns and Peterson~\cite{BurnsP89} gave a solution of time complexity  $\cO(Nk^2)$, for $M=2k-1$ and $r=\cO(N^2)$.
Moir and Anderson~\cite{MoirA95}  improved the time to $\cO(Nk)$, for $M=\frac{k(k+1)}{2}$ and $r=\cO(Nk^2)$.
Further improvements were due to Buhrman et al.~\cite{BuhrmanGHM95}, who achieved $\cO(k^3)$ time, for $M=\frac{k(k+1)}{2}$ and the  number of registers $r=\cO(k^4\min\{3^k,N\})$, and to Moir and Garay~\cite{MoirG96} whose algorithm achieved $\cO(k^2)$ time complexity, for $M=\frac{k(k+1)}{2}$ and $r=\cO(k^3)$ registers, and who also gave another solution with $\cO(k^4)$ time, for $M=2k-1$ and $r=\cO(k^4)$.	

For other work on renaming, see the papers by Afek et al.~\cite{AfekAFST99,AfekBT00,AfekST02}, Brodsky et al. \cite{BrodskyEW06}, and Eberly et al.~\cite{EberlyHW98}.
Randomized renaming algorithms were considered by Alistarh et al.~\cite{AlistarhACGZ11, AlistarhAGW13, AlistarhAGGG10}.
Lower bounds on renaming were given by Alistarh et al.~\cite{AlistarhACGG14}, Attiya et al.~\cite{AttiyaCHP19,AttiyaP16,AttiyaR02}, Burns and Peterson~\cite{BurnsP89}, Casta\~neda et al.~\cite{CastanedaHR14}, Casta\~neda and Rajsbaum~\cite{CastanedaR10,CastanedaR12}, and Helmi et al.~\cite{HelmiHW14}.

Previous work on the problem Store\&Collect includes papers by Afek et al.~\cite{AfekST99}, Attiya et al.~\cite{AttiyaFK04,AttiyaF03,AttiyaFG02,AttiyaKPWW06}, Chlebus et al.~\cite{ChlebusKS-STOC04} and Saks et al.~\cite{SaksSW91}.
Previous work on models with infinite arrivals of processes and infinitely many shared registers includes~\cite{Aguilera04,AspnesSS02,ChocklerM05,GafniMT01,MerrittT03,MerrittT13}.

\section{Technical Preliminaries}

\label{sec:technical-preliminaries}

Algorithms are executed in an asynchronous system with $n$ processes prone to crashes and a set of read-write registers.
Each process~$p$ is identified by its \emph{original name} \texttt{name}$_p$, which is a unique number in some range of names $[N]=\{1,\ldots,N\}$.

If a parameter of a distributed system can be used in a code of algorithm then this parameter is \emph{known}.
In particular, each process~$p$ knows its original name, which is referred to by a specialized variable in codes of algorithms, say \texttt{name}, with process~$p$ substituting \texttt{name}$_p$ as its private value.
We assume throughout that the number $n$ is known.

The following is a standard terminology regarding delays of enabled operations; 
see~\cite{Attiya-Welch-book2004,HerlihyKozlovRajsbaum-book,Lynch-book96} for expositions of these and related concepts.
When, for any configuration in an execution, some process will eventually complete an invoked operation, then the executed algorithm is \emph{non-blocking.}
When, for any configuration in an execution, each process will eventually complete an invoked operation, even when all the remaining processes have crashed, then the algorithm is \emph{wait-free}.


\Paragraph{Competing for registers.}

We introduce a procedure used by processes to compete for a shared register.
The procedure  has two two properties.
First is that a lack of contention guarantees wining.
This means that if there is exactly one process~$p$ working to win a register~$R$, then $p$ eventually wins~$R$.
The second property is that a win provides exclusivity.
This means that once some contender wins a register~$R$, then no other contender will ever win~$R$.
This specification does not require a register to be won by a process when there are multiple contenders but also does not preclude this.

To implement a competition for a register~$R$, we use an auxiliary dedicated shared register~$H_R$ initialized to \texttt{null}.
This register~$H_R$ is used as a placeholder to store a reservation for~$R$.
A pseudocode of a procedure for $p$ is given in Figure~\ref{fig:procedure-competition}.


\begin{figure}[t]
\hrule

\FF

\textsf{procedure} \textsc{Compete-For-Register\,$(R)$}

\FF

\hrule

\FF

\begin{enumerate}[nosep]
\item
read: $\texttt{contention}_p\leftarrow R$

\item
\texttt{if} $\texttt{contention}_p = \texttt{null}$ \texttt{then} write $R\leftarrow \texttt{name}_p$ \texttt{else} \texttt{exit}

\item
read: $\texttt{contention}_p\leftarrow H_R$

\item
\texttt{if} $\texttt{contention}_p = \texttt{null}$ \texttt{then} write $H_R\leftarrow    \texttt{name}_p$ \texttt{else} \texttt{exit}

\item
read: $\texttt{contention}_p\leftarrow R$

\item
\texttt{if} $\texttt{contention}_p = \texttt{name}_p$ \texttt{then return win else exit}

\end{enumerate}

\FF

\hrule

\FF

\caption{\label{fig:procedure-competition}
Pseudocode for a process~$p$ with name \texttt{name}$_p$ to win a shared register $R$. 
It uses a private variable~$\texttt{contention}_p$ and a shared register $H_R$ associated with $R$. Both registers $R$ and $H_R$ are initialized to \texttt{null}.}
\end{figure}


\begin{lemma}
\label{lem:compete-for}

Procedure \textsc{Compete-For-Register$()$} is an implementation of competition for a register.
\end{lemma}

\begin{proof} 
To show correctness, consider two cases corresponding to the specification during a competition for a register~$R$.
If a process~$p$ acts as the only contender to win~$R$, then $p$  eventually writes the value $\texttt{name}_p$ to both registers $H_R$ and $R$, and the final  read from $R$ makes process~$p$ a winner.
Now suppose some process~$p$ wins~$R$ in the presence of a contender process~$q$.
If the first read of~$R$ by~$q$ does not return \texttt{null}  then $q$ exits immediately and no longer contends to win~$R$. 
Otherwise, process~$q$ reads \texttt{null} from~$R$ as its first action, which means that  process~$q$ managed to read from register~$R$ before process~$p$ wrote to~$R$.
Still process~$p$ wins $R$, so $p$ managed to write \texttt{name}$_p$ to~$R$ and then the same value to~$H_R$ and then check that \texttt{name}$_p$ is still in register~$R$, as all this is required to win~$R$.
When process~$p$ confirms by the second read of~$R$ that $\texttt{name}_p$ is still there then $H_R$ already stored $\texttt{name}_p$.
So process~$q$ could overwrite the value \texttt{name}$_p$ at register~$R$ only after process~$p$ read it for the second time, which means after register~$R$ has the value \texttt{name}$_p$ stored in it already.
When process~$q$ reads register~$H_R$ then this occurs after process~$p$ wrote $\texttt{name}_p$ to it, so the read returns $\texttt{name}_p$, which is different from \texttt{null}. 
This makes process~$q$ exit without invoking a write to~$H_R$.
\end{proof}


\Paragraph{Graphs.}

Let $G=(V,W,E)$ be a simple bipartite graph.
This notation means that the vertices are partitioned into the set of \emph{inputs~$V$}  and the set of \emph{outputs~$W$}, $E$ is the set of edges, and each edge has one endpoint in $V$ and the other in~$W$.
We say that graph $G$ has \emph{input-degree $\Delta$} if each vertex in~$V$ is connected to exactly $\Delta$ neighbors in~$W$. 
A graph~$G$ is said to be an \emph{$(L,\Delta,\varepsilon)$-lossless expander}, for a natural number $L$, if $\Delta$ is the input-degree of $G$ and each subset $X$ of $V$ of size $|X|\le L$ has at least $(1-\varepsilon)\Delta\cdot |X|$ neighbors in~$W$.
A vertex $v\in W$ is a \emph{unique neighbor of set $S\subseteq V$} if vertex~$v$ is adjacent to exactly one vertex in~$S$.


\begin{lemma}[\cite{CapalboRVW02}]
\label{lem:from-CapalboRVW02}

Let $G$ be a $(L,\Delta,\varepsilon)$-lossless expander, for some parameters $L$ and $\varepsilon<\frac{1}{2}$.
Then, for each subset $X\subseteq V$ of size $|X|\le L$, at least the fraction $(1- 2\varepsilon )\Delta$ of vertices among the neighbors of~$X$ are unique neighbors.
\end{lemma}

\begin{proof}
Fix an ordering of the sets of vertices~$W$.
For $v\in V$, let $(v,i)$ denote the $i$th neighbor of~$v$ in~$W$.
There are $\Delta |X|$ pairs of the form $(v,i)$, for $v\in X$ and $1\le i\le \Delta$.
By the definition of lossless expanders, all these pairs determine at least $(1-\varepsilon)\Delta |X|$ neighbors of $X$.
It follows that at most $\varepsilon\Delta |X|$ such pairs denote vertices repeated at least twice. 
Each such a repetition $(v_1,i)=(v_2,j)$ eliminates two possible unique neighbors of $X$ per one real neighbor $(v_1,i)=(v_2,j)$, in the total number $\Delta |X|$ of pairs of the form $(v,i)$, for $v\in X$ and $1\le i\le \Delta$.
\end{proof}


\begin{lemma}
\label{lem:partial-matching}

If a bipartite graph $G=(V,W,E)$ is an $(L,\Delta,\varepsilon)$-lossless expander, for some numbers~$L$ and $\varepsilon<\frac{1}{2}$, then, for each set $X\subseteq V$ of size $|X|\le L$, there is a partial matching in~$G$, between some vertices in $X$ and  unique neighbors of~$X$, that has at least $(1-2\varepsilon) |X|$ edges.
\end{lemma}

\begin{proof}
By Lemma~\ref{lem:from-CapalboRVW02}, a fraction of at least $(1- 2\varepsilon )\Delta$ vertices among the neighbors of~$X$ are unique neighbors.
We can match these inputs to their unique neighbors.
\end{proof}

Let $\lg z$ denote the logarithm of $z$ to the base~$2$, and $e$ be the base of natural logarithms.
We will use the existence of lossless expanders with the following properties:


\begin{lemma}
\label{lem:expander}

Let  $V$ and $W$ be two finite disjoint sets and $L$ a natural number such that $1\le L \le \frac{|V|}{2}$.
If $|W|=12e^4 L\lg\frac{|V|}{L}$ then there exists a bipartite graph $G=(V,W,E)$ of input-degree $\Delta$ such that $4\le \Delta\le 4\lg |V|$ and $G$ is a $(L,\Delta,\frac{1}{4})$-lossless expander.
\end{lemma}

\begin{proof} 
We show that a set of edges between the vertices in $V$ and $W$ selected randomly subject to degree constraints meets the requirements with a positive probability.
More precisely, for each vertex $v\in V$, we select  $\Delta$ neighbors in~$W$ uniformly at random.
Next, we demonstrate that the resulting graph is  a $(L,\Delta,\frac{1}{4})$-lossless-expander with a probability greater than~$0$, where a specific value of $\Delta$ will be determined later.

Let  $X\subseteq V$ be  a subset of $x=|X|\le L$ elements, and $Y\subseteq W$ be a subset of $\frac{3}{4}x\Delta$ elements.
The probability that all the neighbors of $X$ are in the set $Y$ is at most
\[
\biggl(\frac{\binom{|Y|}{\Delta}}{\binom{|W|}{\Delta}}\biggr)^{x} \le
\biggl(\frac{|Y|e}{|W|}\biggr)^{x\Delta} \le
\biggl(\frac{\frac{3e}{4}\cdot x\Delta}{|W|}\biggr)^{x\Delta} 
\ ,
\]
where we used the bounds $(\frac{n}{n})^k\le \binom{n}{k} \le(\frac{en}{k})^k$.
The number of different subsets $X\subseteq V$ of size $x$ is at most
\[
\binom{|V|}{x} \le \biggl(\frac{|V|e}{x}\biggr)^x
\ .
\]
The number of different subsets $Y\subseteq W$ of size $\frac{3}{4}x\Delta$ is at most
\[
\binom{|W|}{\frac{3}{4}x\Delta} \le 
\biggl(\frac{|W|e}{\frac{3}{4}x\Delta}\biggr)^{3x\Delta/4}
\ .
\]
Therefore, the probability that there exists a set $X\subseteq V$ with at most $\frac{3}{4}x\Delta$ neighbors, for a given size~$x$,  is at most 
\begin{equation}
\label{eqn:expander-exist-one}
\biggl(\frac{|V|e}{x}\biggr)^x \cdot 
\biggl(\frac{|W|e}{\frac{3}{4}x\Delta}\biggr)^{3x\Delta/4}
\cdot \biggl(\frac{\frac{3e}{4}x\Delta}{|W|}\biggr)^{x\Delta}
\le
\biggl(\frac{|V|e}{x}\biggr)^x \cdot 
\biggl(\frac{\frac{3e^3}{4}x\Delta}{|W|}\biggr)^{x\Delta/4}
\ .
\end{equation}
We assumed $|W|=12e^4 L\lg\frac{|V|}{L}$.
Now, define $\Delta=4\lg\frac{|V|}{L}$.
Observe that $\Delta\ge 4$, because $|V|\ge 2L$. 
The assumption on $W$ and the specification fo $\Delta$ produce cancellations that give:
\[
\frac{3e^4}{4} \cdot \frac{x\Delta}{W}= \frac{3e^4}{4} \cdot \frac{x4\lg\frac{|V|}{L}}{12e^4 L\lg\frac{|V|}{L}} =\frac{x}{4L}
\ .
\]
To estimate~\eqref{eqn:expander-exist-one} further, observe that the quantity raised to the power~$x$ in~\eqref{eqn:expander-exist-one} equals:
\begin{equation}
\label{eqn:expander-exist-two}
\frac{|V|}{L}\cdot \frac{L}{x} \cdot 
\biggl(\frac{3e^4x\Delta}{4 |W|}\biggr)^{\Delta/4}\\
=
\frac{|V|}{L}\cdot\frac{L}{x} \cdot
\biggl(\frac{1}{4}\cdot \frac{x}{L}\biggr)^{\lg\frac{|V|}{L}}
\ .
\end{equation}
Let us introduce the notation $\alpha=\frac{|V|}{L}$ and $\beta= \frac{L}{x}$.
By the assumptions, we have $\alpha\ge 2$ and $\beta\ge 1$.
The right-hand side of~\eqref{eqn:expander-exist-two} can be represented as
\[
\frac{\alpha \,\beta}{(4\beta)^{\lg \alpha}} = \frac{\beta}{\alpha\, \beta^{\lg \alpha}}
\ .
\]
This quantity is at most $\frac{1}{2}$, for $\alpha \ge 2$ and $\beta\ge 1$.
It follows that the right-hand side of~\eqref{eqn:expander-exist-one} is at most $\frac{1}{2^x}$.

We demonstrated that the probability that some subset $X \subseteq V$ of size $x=|X|\le L$ has at most $\frac{3}{4}x\Delta$ neighbors in~$W$ is at most $\frac{1}{2^x}$.
The probability that some  subset $X\subseteq V$ of size $x=|X|\le L$ has at most $\frac{3}{4}|X|\Delta$ neighbors in~$W$ is at most $\sum_{x=1}^{L} \frac{1}{2^x} < 1$.
By the probabilistic method, there exists a bipartite graph $G=(V,W,E)$ with input-degree $\Delta=4\lg\frac{|V|}{L}$ in which every subset $X\subseteq V$ of size $|X|\le L$ has more than $\frac{3}{4}|X|\Delta$ neighbors in~$W$, where the number $L$ satisfies $L\le \frac{|V|}{2}$.
\end{proof}

\section{Bounded Selection}

We consider the problem of assigning new names to a group of participating processes, from among the total of $n$, known as \emph{renaming}.
Each among $n$ processes holds an \emph{original} name from some range $[N]=\{1,\ldots,N\}$, which is the value of its variable \texttt{name}$_p$.
The number of participating processes is denoted as $k$, where $k\le n$.
These participating processes contend to acquire unique integers in a range~$[M]$  as new names using some~$r$ auxiliary shared registers, where $M<N$. 
The goal is to minimize a number of parameters: the running time and a range $[M]$ of new names, but also the number of auxiliary registers $r$.
Running time means local steps, which is a maximum among the processes of the number of time-steps counted by each process until halting. 

We assume that $n$ is known, but whether $k$ and $N$ are known depends on a precise specification of the renaming problem.
This leads to four variants of the problem of renaming, where either (1)~both $k$ and $N$ are known, or (2)~only $k$ is known, or (3~ only $N$ is known, and finally (4)~with none of $k$ and $N$ known.
The case of both $k$ and $N$ unknown is most challenging.
Algorithms for renaming that do not refer to either $k$ nor $N$ in their codes, and so work for arbitrary unknown values of $k$ and $N$, are called \emph{adaptive}.
We apply the convention to add known parameters among $k$ and $N$ to names of algorithms, so that if a parameter  is missing in the name then this means the parameter is unknown.

Our first goal is to develop a renaming algorithm with both $k$ and $N$ known.
We begin by introducing an auxiliary problem called Majority-Renaming, which is about assigning new names to at least half of some $k$ contending processes.
An algorithm is \emph{$(k,N)$-majority renaming with a bound~$M$ on new names} if at least half of any $k$ contending processes with original names in~$[N]$ acquire unique names in~$[M]$, while the numbers $k$ and~$N$ can be part of code of the algorithm.
We find a solution for Majority-Renaming based on lossless expanders with good unique-neighbors properties.
This becomes a stepping stone to develop solutions for Renaming itself.
Employing a renaming algorithm that relies on some known information, we then argue how to obtain an adaptive solution.
Finally, we discuss how to use the obtained renaming algorithms to solve Store\&Collect.

We begin by presenting an algorithm called \textsc{Majority}$(\ell,N)$, where $N$ is a natural number and $\ell\le \frac{N}{2}$, which is $(\ell,N)$-majority renaming.
The number $M=12e^4\ell\lg\frac{N}{\ell}$ serves as a bound on the range of new names.
We consider a bipartite graph $G=(V,W,E)$, where $V=[N]$, $W=[M]$, and each input-degree is~$\Delta$.
The topology of $G$ is such that $G$ satisfies the properties given in Lemma~\ref{lem:expander}.
The graph $G$ is part of code of the algorithm.
The set $V$ of inputs corresponds to all  original names of processes.
Each output vertex in the set~$W$ represents a possible new name. 

The edges of $G$ determine which registers will be competed for by the processes, using procedure \textsc{Compete-For-Register} given in Figure~\ref{fig:procedure-competition}.
To facilitate this, there are two unique shared registers associated with each output vertex.

Competition to win registers in an execution of algorithm \textsc{Majority}$(\ell,N)$ proceeds as follows.
A process~$p$ with a name in $V=[N]$ uses a vertex  in~$V$ labeled $\texttt{name}_p$.
It begins by attempting to win a register corresponding to the first neighbor of $\texttt{name}_p$ in~$W$.
If $p$ fails then it attempts to win the register of its second neighbor in~$W$.
This continues until $p$ either wins a register or exhausts all the neighbors. 
As soon as $p$ wins a register in~$W$ then $p$ adopts the number of the won register as its new name and exits.
If $p$ fails all $\Delta$ competitions, then $p$ halts without acquiring a new name.


\begin{lemma}
\label{lem:majority}

Algorithm \textsc{Majority}$(\ell,N)$,  for a natural number $N$ and $\ell\le \frac{N}{2}$, is $(\ell,N)$-majority-renaming, where $M=12e^4\ell\lg\frac{N}{\ell}$ is a  bound on new names.
The algorithm operates in~$\cO(\log N)$ local steps and uses $24e^4\ell\lg\frac{N}{\ell}$ auxiliary registers.
\end{lemma}

\begin{proof} 
The graph~$G$ is a $(\ell,\Delta,\frac{1}{4})$-lossless-expander, as stated in Lemma~\ref{lem:expander}.
A majority of contending processes have unique neighbors not shared with other active processes, by Lemma~\ref{lem:partial-matching} applied to graph~$G$.
If an active process has a unique neighbor then it eventually wins some register representing its neighbor, by Lemma~\ref{lem:compete-for}. 
Hence a majority of active processes get unique names.
The worst-case running time is proportional to the degree $\Delta$ of graph~$G$, which is $\cO(\log N)$. 
The number of used registers is~$2M$, as we use two unique shared registers per one vertex in~$W$. 
\end{proof}


\Paragraph{Renaming when both $k$ and $N$ are known.}

Next, we consider an algorithm \textsc{Plain-Rename}$(k,N)$, which is  $(k,N)$-renaming with $M=24e^4k\lg\frac{N}{k}$ as a bound on new names.
A process~$p\in [N]$ proceeds through consecutive stages,  up to $1+\lg k $ stages maximum,  until it gets a unique name. 
In a stage~$i$, where $0\le i\le \lg k$, process~$p$ executes \textsc{Majority}$(\frac{k}{2^i},N)$ on the set of pairs of registers~$M_i$, where $|M_i|=12e^4 \frac{k}{2^i}\lg\frac{N2^i}{k}$.
We assume that the sets $M_i$ are mutually disjoint.
A union of these sets~$\cup_{i=1}^{1+\lg k }M_i$ constitutes a collection of new names.


\begin{lemma}
\label{lem:Plain-Rename}

Algorithm \textsc{Plain-Rename}$(k,N)$ is $(k,N)$-renaming with $M=24e^4k\lg\frac{2N}{k}$ as a bound on new names.
It  operates in~$\cO(\log k\log N)$ local steps and uses $48e^4k\lg\frac{2N}{k}$ auxiliary registers. 
\end{lemma}

\begin{proof}
Procedure \textsc{Majority}$(\frac{k}{2^i},N)$ is majority renaming, by Lemma~\ref{lem:majority}.
Each call of this procedure at least halves the number of contending processes that still need names.
Calls of the procedure continue until there remain at most one process without a new name, which then eventually also gets a name. 
This takes $\cO(\log k \cdot \log N)$ local steps in total,  by Lemma~\ref{lem:majority}.
Number $1+\lg k $ is a bound on the number of stages.
The size of a pool of new names can be estimated as follows: 
\[
M=\sum_{i=0}^{\lg k} 12e^4\frac{k}{2^i} \lg\frac{2^i N}{k} 
= 12e^4 k  \Bigl( \sum_{i=1}^{\lg k} \frac{i}{2^i}+\sum_{i=0}^{\lg k}  \frac{1}{2^i}\lg\frac{N}{k} \Bigr)
=12e^4 k \bigl(2 + 2\lg\frac{N}{k} \bigr)
= 24e^4k\log\frac{2N}{k}
\ .
\]
The number of  shared registers needed for this to work correctly is $2M$.  
\end{proof}

Next, we consider an $(k,N)$-renaming algorithm called \textsc{Compact-Rename}$(k,N)$. 
It is an improvement over \textsc{Plain-Raname} in that the range of new names is~$\cO(k)$.

An execution of \textsc{Compact-Rename}$(k,N)$ is structured as a sequence of epochs.
A process~$p\in [N]$ participates in consecutive epochs, in each one getting a new name.
At least one epoch is executed as long as $N\ge 2^{15}k$, otherwise the original names serve as new names without any change.
The original names are used in the first epoch, and then the names acquired in epoch $i$ are used as original names in epoch $i+1$.  
This continues until the upper bound on the range of the new names, determined by the properties of a current epoch, becomes less than~$2^{15}k$.
A process acquires its ultimate name during this last epoch.

In epoch~$j$, process~$p$ executes \textsc{Plain-Rename}$(k,N_j)$, where $N_1=N$ and $N_{j+1}=24e^4k\lg\frac{2N_j}{k}$ are bounds on new names in calls of \textsc{Plain-Rename}$(k,N_j)$, for $j\ge 1$.
The executions of instantiations of algorithm \textsc{Plain-Rename} use sets of shared registers dedicated for each epoch, such that a shared register is used in only one epoch.
Processes use names from the range $[N_j]$ in epoch $j$,  with $N_1=N$.
The names assigned in epoch~$j$ are from the range $[N_{j+1}]$.

The algorithm \textsc{Compact-Rename}$(k,N)$ terminates because the ranges of new names shrink with passing epochs: $N_{j+1}<N_j$, for $j>0$.
We evaluate the rate of shrinking next.


\begin{lemma}
\label{lem:shrinking-rate}

If $N\ge 2^{15}k$ then $\frac{N_{2}}{N_1}< \frac{2}{3}$, and as long as $N_{j-1}\ge 2^{15}\cdot k$ then $\frac{N_{j+1}}{N_j}<\frac{24}{25}$,  for $j> 1$.
\end{lemma}

\begin{proof}
The case of $j=1$:
\[
\frac{N_2}{N_1}
= \frac{24\,e^4\,k\lg\frac{2N}{k}}{N}
\le 
\frac{24\,e^4\,k\lg\frac{2^{16}k}{k}}{2^{15} k}
=
\frac{24\, e^4 \, 16}{2^{15}}
<\frac{2}{3}
\ .
\]
The case of $j>1$:
\[
\frac{N_{j+1}}{N_j }
= \frac{24e^4 k\lg\frac{2N_j}{k}}{24e^4 k\lg\frac{2N_{j-1}}{k}} 
= \frac{\lg(48e^4 \lg\frac{2N_{j-1}}{k})}{\lg\frac{2N_{j-1}}{k}} 
= \frac{\lg(48e^4)}{\lg\frac{2N_{j-1}}{k}} + \frac{\lg\lg\frac{2N_{j-1}}{k}}{\lg\frac{2N_{j-1}}{k}}
< \frac{71}{100} + \frac{1}{4}=\frac{24}{25}
\ ,
\]
for $N_{j-1}\ge 2^{15}\cdot k$.
\end{proof}

Define a sequence $(a_j)_{j\ge 1}$: $a_j= \lg(\frac{2N_j}{k})$ for $j\ge 1$.


\begin{lemma}
\label{lem:sequence-a}

If $N_j\ge 2^{16} k$ then $a_{j+1}<\log_\frac{3}{2} a_j$.
\end{lemma}

\begin{proof}
The sequence $a_j$ satisfies the following recurrence, for $j\ge 1$:
\[
a_{j+1} = \lg\Bigl(\frac{2N_{j+1}}{k}\Bigr)= \lg (48e^4 \lg\Bigl(\frac{2N_j}{k}\Bigr)) = \lg (48e^4a_j)
\ .
\]
The inequality $\lg (48e^4 x)<\log_\frac{3}{2} x$ holds for $x\ge 17$, by inspection.
Substituting $a_j= \lg(\frac{2N_j}{k})$ for $x$, we obtain that $a_{j+1}<\log_\frac{3}{2} a_j$ for $\lg(\frac{2N_j}{k})\ge 17$, which holds for $N_j\ge 2^{16} k$.
\end{proof}

We use the iterated-logarithm function $\log^{(i)} n$, which denotes $\log n$ iterated $i$ times.
A recursive definition of this function reads $\log^{(0)} n=n$ and $\log^{(i+1)} n= \log \log^{(i)} n$.
We also refer to $\log^\ast n =\min\{i: \log^{(i)} n\le 1\}$, for $n>1$.


\begin{theorem}
\label{thm:compact-rename}

Algorithm \textsc{Compact-Rename}$(k,N)$ is $(k,N)$-renaming with $M=2^{15}\, k$ as a bound on new names,  assuming $N\ge 2^{15}\,k$.
It operates in $\cO(\log k (\log N + \log k\log^* \frac{N}{k}))$ local steps and uses $\cO(k\log\frac{N}{k})$ auxiliary registers. 
\end{theorem}

\begin{proof}
We assume  that $N\ge 2^{15}k$, as otherwise no epoch is executed.
The ranges of names used through the epochs can be traced back to $N_1=N$ as follows.
If $N_j\ge 2^{16} k$ then 
\[
N_{j+1} 
=
24e^4 k \lg\frac{2N_j}{k} 
=
24e^4 k a_j
< 
24e^4 k \log_{\frac{3}{2}}^{(j-i)} a_i
\ ,
\]
as long as $N_{j-i}\ge 2^{16}k$, by Lemma~\ref{lem:sequence-a}.
Combining this with Lemma~\ref{lem:shrinking-rate}, we obtain the bound
\[
N_{j+1} = \cO\Bigl(k\log^{(j)} \frac{N}{k}\Bigr)
\ ,
\]
 for $N\ge 2^{15}k$.
The final range of new names $[N_{j^\ast+1}]$ satisfies $N_{j^*+1} < 2^{15} k$.
The number of epochs is $j^\ast=\cO(\log^\ast \frac{N}{k})$.

The number of local steps can be estimated as follows:
\[
\cO\bigl(\sum_{j=1}^{j^*} \log k \log N_j\bigr) 
=
\cO\bigl(\log k \sum_{j=1}^{j^*}  (\log k + \log^{(j)} N)\bigr) 
=
\cO(\log^2 k\cdot \log^* \frac{N}{k}+\log k \cdot \log N )\ ,
\]
by Lemma~\ref{lem:Plain-Rename} and the bound  on the number of epochs  $j^\ast=\cO(\log^\ast \frac{N}{k})$.

The first epoch uses $48 e^4\lg \frac{2N}{k}$ registers.
A number of registers used in subsequent epochs is given by Lema~\ref{lem:Plain-Rename}.
These numbers keep decreasing, with a rate determined Lemmas~\ref{lem:shrinking-rate} and~\ref{lem:sequence-a}.
By combining this together we obtain that the number of needed registers  is
\[
\sum_{j=1}^{j^*+1} 48 e^4\lg \frac{2N_j}{k} = 48 e^4 k \sum_{j=1}^{j^*+1} a_j
 \ .
\]
The estimate on the rate of decreasing of $a_j$ given in Lemma~\ref{lem:sequence-a} applies for all but $\cO(1)$ epochs.
It follows that the number of needed registers is $\cO(k a_1)= \cO(k\log\frac{N}{k})$.
\end{proof}

The knowledge of $k$ and range $N$ that is polynonomial in $n$ allows to obtain a renaming algorithm efficient with respect to the total number of processes~$n$.


\begin{corollary}
\label{cor:known-k-and-N-rename}

If $k$ and $N$ are known and $N=\cO(n)$, then algorithms \textsc{Compact-Rename}$(k,N)$ runs in $\cO(\log^2 n)$ local steps and uses $\cO(n)$ auxiliary registers. 
If $k$ and $N$ are known and $N$ is polynomial in $n$, then algorithm \textsc{Compact-Rename}$(k,N)$ runs in $\cO(\log^2 n \log^* n)$ local steps and uses $\cO(n\log n)$ auxiliary registers. 
\end{corollary}

\begin{proof}
Algorithm \textsc{Compact-Rename}$(k,N)$ needs $\cO(\log k (\log N + \log k\log^* \frac{N}{k}))$ local steps, by Theorem~\ref{thm:compact-rename}.
This bound is $\cO(\log^2 n)$ if  $N=\cO(n)$, and it is $\cO(\log^2 n \log^* n)$  if $N$ is polynomial in~$n$.
Algorithm \textsc{Compact-Rename}$(k,N)$ uses $\cO(k\log\frac{N}{k})$ auxiliary registers, by Theorem~\ref{thm:compact-rename}.
If a known~$N$ is such that $N=\cO(n)$ then use $\cO(k\log\frac{n}{k})=\cO(n)$, which holds for $1\le k\le n$, to obtain the bound $\cO(n)$ on the number of registers.  
If a known $N$ satisfies $N=\cO(n^\alpha)$, for $\alpha> 1$, then the bound on the number of registers becomes $\cO(k\log\frac{N}{k})=\cO(n\log n)$.
\end{proof}


\Paragraph{Renaming when $N$ is known while $k$ is not.}

We present an algorithm  \textsc{Range-Rename}$(N)$, which renames $k$ contending processes assigning names of magnitude $\cO(k)$.
A bound~$N$ on the magnitude of original names is known but the number of participating processes $k$ is unknown. 

The algorithm is specified as follows.
A process participates in consecutive executions of algorithms \textsc{Compact-Rename}$(2^i,N)$, for consecutive integers $i=1,2,3\ldots$, until it obtains a new name. 
After the first execution is over, it starts a new execution only after the previous one has failed to assign a new name.
These consecutive executions use disjoint  sets of registers.
An execution of \textsc{Compact-Rename}$(2^i,N)$, for $i>1$, uses a next  interval of integers as a range of new names, following the intervals used by \textsc{Compact-Rename}$(2^j,N)$, for $1\le j<i$. 
This means the size of interval used by \textsc{Compact-Rename}$(2^i,N)$ is $2^{15} 2^i$, as indicated by Theorem~\ref{thm:compact-rename}, but it is of the form $[\ell, \ell + 2^{15} 2^i]$, for a suitable positive integer~$\ell$ that is greater than~$1$, except for $i=1$.


\begin{theorem}
\label{thm:range-rename}

Algorithm \textsc{Range-Rename}$(N)$ is $N$-renaming.
It assigns new names of magnitude $2^{16}k$, runs in $\cO(\log^2 k (\log N + \log k\log^\ast N))$ local steps, and uses $\cO(n\log\frac{N}{n})$ auxiliary registers.
\end{theorem}

\begin{proof}
At most $k$ processes participate in each execution of \textsc{Compact-Rename}$(2^i,N)$. We have that  $i\le \lceil \lg k\rceil$, because as soon as $k\le 2^i$ then each process acquires a new name.
The size of the range of new names is estimated by Theorem~\ref{thm:compact-rename} to be at most 
\[
2^{15}\sum_{j=1}^{\lceil \lg k\rceil} 2^j\le 2^{16} k
\ .
\]
The number of steps through executing  \textsc{Compact-Rename}$(2^{\lceil\lg k\rceil},N)$ follows from Theorem~\ref{thm:compact-rename}:
\[
\cO\Bigl(\sum_{j=1}^{\lceil\lg k\rceil} \Bigl(j\log N +j^2\log^* \frac{N}{2^j}\Bigr) \Bigr) =
\cO(\log^2 k (\log N + \log k \log^\ast N))
\ .
\]
The number of registers follows from Theorem~\ref{thm:compact-rename}:
\[
\cO\left(\sum_{j=1}^{\lceil\lg n\rceil} 2^j\log\frac{N}{2^j}\right) =\cO\left(n\log\frac{N}{n}\right)
\ , 
\]
 since $k\le n$.
\end{proof}

The renaming algorithms for the case when $N$ is known use few auxiliary registers if $N$ is   polynomial in~$n$.


\begin{corollary}
\label{cor:only-known-N-rename}

If a known range of the original names~$N$ is polynomial in $n$, then algorithm   \textsc{Range-Rename}$(N)$ runs in $\cO(\log^3 n \log^\ast n)$ local steps and uses $\cO(n\log n)$ auxiliary registers. 
If a known range of the original names $N$ satisfies $N=\cO(n)$, then algorithm  \textsc{Range-Rename}$(N)$ uses  $\cO(n)$ auxiliary registers. 
\end{corollary}

\begin{proof}
Algorithm \textsc{Range-Rename}$(N)$ runs in $\cO(\log^2 k (\log N + \log k\log^\ast N))$ local steps, by Theorem~\ref{thm:range-rename}.
The time bound becomes $\cO(\log^3 n \log^\ast n)$ if $N$ is polynomial in~$n$, since $k\le n$.
Algorithm \textsc{Range-Rename}$(N)$ uses $\cO(n\log\frac{N}{n})$ auxiliary registers, by Theorem~\ref{thm:range-rename}.
If $N=\cO(n)$ then the number of registers is $\cO(n\log\frac{N}{n})=\cO(n)$, and if $N=\cO(n^\alpha)$, for $\alpha\ge 1$, then the number of registers is $\cO(n\log\frac{N}{n})=\cO(n\log n)$.
\end{proof}


\Paragraph{Renaming when $k$ is known while $N$ is not.}

We design an algorithm \textsc{Square-Rename}$(k)$, where a range $N$ is unspecified.
The algorithm assigns new names with $2k-1$ as a range of new names.

The algorithm refers to three other algorithms, some for the case when both parameters $k$ and $N$ are known.
Let \textsc{MA}$(k)$ be an adaptive algorithm given by Moir and Anderson~\cite{MoirA95}, which is $k$-renaming with $M=\cO(k^2)$ as a  bound on new names.
It operates in $\cO(k)$ local steps and uses $\cO(k^2)$ auxiliary registers.
Let \textsc{AF}$(k,N)$ be the algorithm of Attiya and Fouren~\cite{AttiyaF01}, which is $(k,N)$-renaming with $2k-1$ as a bound on new names.
It operates in $\cO(N)$ local steps and uses $\cO(N^2)$ auxiliary registers.
We use algorithm \textsc{Compact-Rename} together with algorithms \textsc{AF} and \textsc{MA} to obtain a new algorithm called \textsc{Square-Rename}$(k)$, which is $k$-renaming with $2k-1$ as a bound on new names, for any $k$ and $N$.

Algorithm \textsc{Square-Rename}$(k)$ is structured into three parts.
The sets of registers used in all three parts are disjoint.
First run algorithm \textsc{MA}$(k)$ to rename, with $C k^2$ as a  bound on new names, for some $C>0$.
Continue by invoking \textsc{Compact-Rename}$(k,Ck^2)$ to rename again, with $2^{15} k$ as a  range of new names, by Theorem~\ref{thm:compact-rename}.
The processes execute \textsc{Compact-Rename}$(k,Ck^2)$ with the names obtained in the preceding execution of \textsc{MA}$(k)$ treated as original names. 
Finally, execute \textsc{AF}$(k, 2^{15} k)$ to rename one more time.
The processes execute \textsc{AF}$(k,2^{15}k)$ with names obtained in \textsc{Compact-Rename}$(k,Ck^2)$ treated as original names. 
The final new names are as assigned by \textsc{AF}$(k,2^{15} k)$.


\begin{theorem}
\label{thm:square-rename}

Algorithm \textsc{Square-Rename}$(k)$ is $k$-renaming with $2k-1$ as a bound on new names.
It operates  in~$\cO(k)$ local steps and uses $\cO(k^2)$ auxiliary registers.
\end{theorem}

\begin{proof}
We rely on the properties of algorithm~\textsc{MA} from~\cite{MoirA95}, and on the properties of  algorithm~\textsc{AF} given in~\cite{AttiyaF01}.
Correctness follows from the fact that each of the three renaming algorithms properly handles original names, possibly yielded by a preceding execution, assuming they are given correct  ranges of original names, if they rely on  this information.
The range of names is reduced first  to $Ck^2$ by algorithm~\textsc{MA}, for a suitable $C>0$.
Then algorithm~\textsc{Compact-Rename} takes over, with $N=Ck^2$, and reduces the range of names to $2^{15} k$, by Theorem~\ref{thm:compact-rename}. 
Finally, algorithm~\textsc{AF} reduces the range of names to $2k-1$, while using $N=2^{15} k$.

The local step complexity of algorithm \textsc{Square-Rename}$(k)$ is  
\[
\cO\left(k+\log k \left(\log k^2 + \log k\log^* \frac{k^2}{k}\right) +k\right)=\cO(k)
\ , 
\]
by the properties of  algorithm \textsc{MA} in~\cite{MoirA95}, and by these of  algorithm~\textsc{AF} given in~\cite{AttiyaF01}, and by Theorem~\ref{thm:compact-rename}.
The number of needed registers is at most
\[
\cO\left(k^2+k\log\frac{k^2}{k}+k^2\right)=\cO(k^2)
\ ,
\]
by the respective properties of algorithms \textsc{MA}$(k)$ and \textsc{AF}$(k,2^{15}k)$, and by Theorem~\ref{thm:compact-rename}.
\end{proof}


\Paragraph{Adaptive renaming with both $k$ and $N$ unknown.}

Now we develop algorithm \textsc{Adaptive-Rename} solving Renaming  in a fully adaptive fashion.
The algorithm uses \textsc{Square-Rename} as a subroutine.
It operates as follows.
A process~$p$ executes instantiations of algorithm \textsc{Square-Rename}$(2^j)$, for consecutive integers $i=0,1,2,\ldots$. 
If $p$ does not obtain a new name in an execution of \textsc{Square-Rename}$(2^j)$, then it proceeds to execute \textsc{Square-Rename}$(2^{j+1})$.
Once a process~$p$ acquires a new name, then it halts.
Executions of \textsc{Square-Rename}$(2^{j})$ for different values of $j$ use dedicated sets of auxiliary registers assigned to each possible integer value~$j\le \lceil\lg k\rceil$ and also dedicated ranges of names assigned to $j$.
These ranges are disjoint, and a range for $j+1$ immediately follows a range for $j$, such that the range for $i$ such that $0\le i\le j$ fill a contiguous segment.
Algorithm \textsc{Adaptive-Rename} does  not have the parameters $k$ and $N$ in its code.


\begin{theorem}
\label{thm:adaptive-rename}

Algorithm \textsc{Adaptive-Rename}  solves Renaming  in an adaptive manner. 
The range of new names is $M=8k-\lg k-1$.
The number of local steps is~$\cO(k)$ and the number of auxiliary registers is~$\cO(n^2)$.
\end{theorem}

\begin{proof}
Consider an instantiation of \textsc{Square-Rename}$(2^j)$ for $j=\lceil\lg k\rceil$, which is the latest possible. 
At most $k$ processes participate in this execution.
By Theorem~\ref{thm:square-rename}, the magnitude  of new names is bounded above by 
\[
\sum_{i=0}^{\lceil\lg k\rceil} (2^{i+1}-1) = 2^{\lceil\lg k\rceil+2}-(\lceil\lg k\rceil+1) \le 8k-\lg k -1
\ .
\]
The number of local steps of each process is bounded from above by 
\[
\cO\Bigl(\sum_{i=1}^{\lceil\lg k\rceil} 2^i\Bigr) =\cO(2^{\lceil\lg k\rceil}) =\cO(k)
\ , 
\]
by Theorem~\ref{thm:square-rename}.
The number of needed auxiliary registers used is  
\[
\cO\Bigl(\sum_{i=1}^{\lceil\lg k\rceil} 2^{2i}\Bigr) = \cO(2^{2 \lceil\lg k\rceil}) = \cO(n^2)
\ , 
\]
again by Theorem~\ref{thm:square-rename}. 
\end{proof}

There is an alternative adaptive solution for Renaming, which works as follows.
First execute an adaptive version of  algorithm \textsc{MA}.
It accomplishes renaming in $\cO(k)$ local steps and $k^2$ new names using $\cO(n^2)$ registers.
Follow by executing algorithm \textsc{Range-Rename}$(k^2)$.
By Theorem~\ref{thm:range-rename}, this solves Renaming  in $\cO(k+\log^3 k)=\cO(k)$ local steps while using $\cO(n^2+n\log n)=\cO(n^2)$ registers.
A drawback of this algorithm is that the range of new names, although still $\cO(k)$, has a large constant factor by~$k$.
This is not the case for the algorithm \textsc{Adaptive-Rename}, where the range of names is smaller than $8k-\lg k$.


\Paragraph{Solutions for Store\&Collect.}

We show how to implement the \texttt{Store} and \texttt{Collect} operations, with $k$ processes out of $n$ participating, with original names in the range~$[N]$.
We want the first \texttt{Store} by a process to begin by executing a suitable renaming algorithm.
A new name identifies a register which the process uses to store by writing into it.
Each of the subsequent calls of \texttt{Store} takes a constant number of local steps, because the register to write to has already been identified 
The algorithms for Store\&Collect are categorized by the levels of knowledge of $k$ and $N$, with $n$ always known.

A number~$M$ is a range of new names in a renaming algorithm.
If $M$ can be computed in advance, based on the available knowledge of the numbers $k$ and $N$, then an algorithm can be structured as follows.
We allocate $M$ shared read-write registers $S_i$ indexed by natural numbers~$i$ in~$[M]$.
In order to store a value, a process first seeks a new name~$x$ in~$[M]$.
This assigns a unique read-write register~$S_x$.
The process with a new name~$x$ stores its values by writing into~$S_x$.
Collecting is performed by reading the registers~$S_i$, for all $i\in [M]$.
This approach works when both parameters~$k$ and~$N$ are known.


\begin{theorem}
\label{thm:collect-k-N-both-known}

If both $k$ and $N$ are known, then Store\&Collect can be implemented such that a first storing operation by a process takes  $\cO(\log k (\log N + \log k\log^* \frac{N}{k}))$ steps and collecting $\cO(k)$ steps, while using $\cO(k\log\frac{N}{k})$ auxiliary registers.
\end{theorem}

\begin{proof}
Since we use algorithm \textsc{Compact-Rename}$(k,N)$ for renaming, it is Theorem~\ref{thm:compact-rename} that summarizes the performance characteristics of renaming. 
The number of registers to store values corresponding to the new names is~$M=2^{15}k$.
The first \texttt{Store} operation of a process is preceded by acquiring a new name.
This takes  $\cO(\log k (\log N + \log k\log^* \frac{N}{k}))$ local steps.
The \texttt{Collect} operation consists of reading all the $2^{15} k$ shared registers, which takes time $\cO(k)$.
The algorithm uses $\cO(k\log\frac{N}{k})$ registers for renaming and $2^{15}k$ registers for storing, which together make $\cO(k\log\frac{N}{k})$ auxiliary registers.
\end{proof}

The knowledge of $k$ and a range of original names~$N$ that is polynomial in~$n$ allows to obtain a Store\&Collect algorithm efficient with respect to the total number of processes~$n$.


\begin{corollary}
\label{cor:store-collect-both-k-and-N-known}

If both $k$ and $N$ are known and  $N=\cO(n)$ then Store\&Collect can be implemented such that a first storing operation by a process takes  $\cO(\log^2 n)$ local steps  and collecting $\cO(k)$ steps, while using $\cO(n)$ auxiliary registers.
If both $k$ and $N$ are known and  $N$ is polynomial in~$n$ then Store\&Collect can be implemented such that a first storing operation by a process takes  $\cO(\log^2 n \log^* n)$ local steps  and collecting $\cO(k)$ steps, while using $\cO(n\log n)$ auxiliary registers.
\end{corollary}

\begin{proof}
This is a specialization of the solution of Store\&Collect when both $k$ and $N$ are known with performance summarized in Theorem~\ref{thm:collect-k-N-both-known}.
Since we use algorithm \textsc{Compact-Rename}$(k,N)$ for renaming, we refer to Corollary~\ref{cor:known-k-and-N-rename} that summarizes it performance characteristics with additional assumptions on the magnitude of~$N$.
\end{proof}

Next, we consider the case when $N$ is known but $k$ is unknown.
These assumptions qualify algorithm \textsc{Range-Rename}$(N)$ for renaming to be applicable.
This algorithm assigns new names that are in the range between~$1$ and~$2^{16} k$, if $k$ processes contend for new names, by Theorem~\ref{thm:range-rename}.
A solution to Store\&Collect needs to have at least~$2^{16} n$ registers allocated in advance, since the unknown $k$ could be as large as~$n$.

In order to read only the registers actually used for storing, while performing \texttt{Collect}, we suitably mark an initial range of registers during writing into them not to waist time reading unused registers.
This is implemented as follows.
The number~$2^{16} n$ of registers needs to be incremented by $16+\lceil \lg n\rceil$. 
These many registers are identified by the integers in the range $[1,2^{16} n + \lceil \lg n\rceil+16]$ as their \emph{primary addresses}, in that the $i$th register has primary address~$i$.
All these registers are initialized to be null as usual. 
The registers with primary addresses of the form $i+2^i$, for $i\ge 0$, play an \emph{auxiliary} role.
If a value written in such a register is different from null then the register is \emph{marked}. 
The registers with primary addresses between $2^i+i+1$ up to $2^{i+1}+i$ have secondary addresses between $2^i$ and $2^{i+1}-1$, for $i\ge 0$.
Once can verify directly that each register is either auxiliary or it has a unique secondary address, and each integer between $1$ and $2^{16}n$ is a secondary address of a unique register. 

A process with a new name $m$ stores its proposed value along with its original name in a register of the secondary address~$m$.
During its first operation \texttt{Store}, the process marks all the auxiliary registers whose primary addresses are less than the primary address of the register with $m$ as the secondary address.
In order to perform \texttt{Collect}, a process reads the consecutive registers up to the first unmarked auxiliary register or all the registers if all auxiliary registers are marked.
An unmarked auxiliary register indicates that no value has been stored beyond this primary address.

\begin{theorem}
\label{thm:collect-only-N-known}

If $N$ is known but $k$ is not, then Store\&Collect can be implemented such that a first instance of storing takes $\cO(\log^2 k (\log N + \log k\log^\ast N))$ local steps and collecting $\cO(k)$ steps, with $\cO(n\log\frac{N}{n})$ auxiliary registers.
\end{theorem}

\begin{proof}
As a renaming subroutine, we use  \textsc{Range-Rename}$(N)$, and refer to Theorem~\ref{thm:range-rename}, which summarizes the performance characteristics of the algorithm. 
The number of registers to store values corresponding to the new names is~$2^{16}n$, because the range of new names is $M=2^{16}k$ with $k\le n$.
The first \texttt{Store} operation of a process takes time $\cO(\log^2 k (\log N + \log k\log^\ast N))$ to identify a new name and then up to $1+\lg k$ steps to mark auxiliary registers and store a proposed value.
Each of the subsequent \texttt{Store} operations takes constant time.
The \texttt{Collect} operation consists of reading up to $\lceil \lg k\rceil$ auxiliary registers and up to the $2^{16} k$ registers with the smallest secondary addresses.
The algorithm uses $\cO(n\log\frac{N}{n})$ auxiliary registers for renaming and $2^{16} n + \lceil \lg n\rceil+16$ registers for storing, which together make $\cO(n\log\frac{N}{n})$ registers.
\end{proof}

The  algorithms for Store\&Collect for the case when $N$ is known use few auxiliary registers if $N$ is  polynomial in~$n$.


\begin{corollary}
\label{cor:store-collect-N-known}

If a known range of the original names $N$ is such that $N=\cO(n)$, then Store\&Collect can be implemented using $\cO(n)$ auxiliary registers.
If a known range of the original names $N$ is polynomial in $n$, then Store\&Collect can be implemented using $\cO(n\log n)$ auxiliary registers. 
In each of these cases, a first operation of storing can be performed in a number of local steps poly-logarithmic in~$n$ and collecting takes $\cO(k)$ steps.
\end{corollary}

\begin{proof}
This is a specialization of the solution of Store\&Collect when $N$ is known but $k$ is not with performance summarized in Theorem~\ref{thm:collect-only-N-known}.
Since we use algorithm \textsc{Range-Rename}$(N)$ for renaming, we refer to Corollary~\ref{cor:only-known-N-rename} that summarizes it performance characteristics with additional assumptions on the magnitude of~$N$.
\end{proof}

Next we consider the case when $k$ is known but $N$ is not.
We want to use the algorithm \textsc{Square-Rename}$(k)$ for renaming.
By Theorem~\ref{thm:square-rename}, the range of new names can be determined as $M = 2k-1$ and renaming requires $\cO(k^2)$ auxiliary registers.
A solution to Store\&Collect needs $Ck^2$ shared registers for renaming, for a suitable $C>0$, and $2k-1$ registers for storing.


\begin{theorem}
\label{thm:collect-only-k-known}

If the number of participating processes~$k$ is known but the range of original names~$N$ is not, then Store\&Collect can be implemented such that a first instance of storing takes $\cO(k)$ local steps and collecting $\cO(k)$ steps, with $\cO(k^2)$ auxiliary registers.
\end{theorem}

\begin{proof}
The design of the implementations of \texttt{Store} and \texttt{Collect} follows the general approach of having a dedicated block of $2k-1$ registers for storing and each participating process first acquiring a new name to identifies a register.
By Theorem~\ref{thm:square-rename}, this takes $\cO(k)$ local steps while using $\cO(k^2)$ auxiliary registers.
To collect the proposed values to build a view, a process reads all $2k-1$ registers used for storing, which takes $\cO(k)$ local steps.
\end{proof}

Next we consider the adaptive case with both $k$ and $N$ unknown.
We want to use algorithm \textsc{Adaptive-Rename} for renaming. 
According to its performance characteristics summarized in Theorem~\ref{thm:adaptive-rename}, the range of new names is $M=8k-\lg k-1$ and the number of auxiliary registers is~$\cO(n^2)$.
Registers used for storing can be handled similarly as in the case of unknown $k$ and known~$N$, by arranging a suitably large segment of registers and using them by referring to primary and secondary addresses.
We need $8n -\lg n-1$ secondary addresses, so a block of registers with $8n -\lg n-1 + \lceil \lg(8n -\lg n-1) \rceil$ primary addresses suffices.
Additionally, we need to allocate $Dn^2$ registers for renaming, for a suitable $D>0$.

\begin{theorem}
\label{thm:collect-3}

If $k$ and $N$ are both unknown, then Store\&Collect can be implemented adaptively such that first storing takes $\cO(k)$ steps and collecting takes $\cO(k)$ steps, with $\cO(n^2)$ auxiliary registers.
\end{theorem}

\begin{proof}
We use  \textsc{Adaptive-Rename} for renaming purposes.
The total number of auxiliary registers is $\cO(n^2)$, with $\cO(n)$ for storing and $\cO(n^2)$ for renaming. 
The first store operation of a process takes $\cO(k+\log k)=\cO(k)$ local steps, by Theorem~\ref{thm:adaptive-rename}.
The subsequent store operations are of constant time each.
The collecting operation takes $\cO(k)$ steps, since there is only a prefix of $\cO(k)$ registers corresponding to the names to be read. 
\end{proof}

An adaptive solution of Store\&Collect with performance as in Theorem~\ref{thm:collect-3} has already been given by Afek and De Levie~\cite{AfekL07}, by using a different approach.

\section{Lower Bounds on Local Steps}

We consider the time complexity of renaming and storing for the first time.
A \emph{configuration} of the system consists of events enabled by each process, which are either reads or writes. 
An \emph{execution} consists of events as they occur in time.
An event that occurs is selected from a current configuration. 
A process participating at a current event immediately enables either a read or a  write to contribute an option for such a read or write to occur in the next configuration.


\Paragraph{A lower bound for renaming.}

For a distributed system of $n$ processes, an instance of Renaming  is determined by some of the following four  parameters: an upper bound on the number of participating processes~$k$, a range of new names~$M$, a range of original names~$N$, and a number of auxiliary shared read-write registers~$r$.
For some configurations of these parameters, the number of steps can be very small; for instance,  if $N=\Theta(k)$ then the original names could do, so the number of steps is~$\cO(1)$. 
On the other side of the spectrum, if $N$ is not known and can be arbitrarily large, then the step complexity of any renaming algorithm is $k-1$.

An intuition why the number of steps is $k-1$ in some scenarios could be build as follows. 
If the original names affect the actions of processes, then there is such an assignment of the $n$ original names that at any point of an execution, if there is a choice that processes can make which is affected by the original names, then all the processes might choose similarly.
In particular, if processes choose not to write, so that there is no communication among them, then all want to assign the same name.
Therefore, one of the processes writes at some point, and let $p$ be the first such a process.
After a write, some processes learn of the value written, by reading. 
All processes have to learn the value written by~$p$, since otherwise one process among these that never learn could choose the same name as~$p$.
This argument can be extended by induction to imply that the process that chooses the name as the last one had to read at least $k-1$ times.
A general lower bound of Jayanti et al.~\cite{JayantiTT00} captures related scenarios.

We present a lower bound which gives an estimate on $N$, depending on fixed range of original names~$M$ and the number of auxiliary shared read-write register~$r$, for which $k-1$ steps during assigning new names are necessary.
The lower bound is flexible enough to be applicable to scenarios when algorithms of poly-logarithmic step complexity exist.


\begin{theorem}
\label{thm:rename-lower-bound}

Let $k$ processes among a total of $n$ in an asynchronous system using $r$ shared read-write registers and with original names in the range $[N]$ execute a wait-free renaming algorithm that assigns new names in a range $[M]$. 
Then there exists an execution in which some process performs $1+\min\{k-2,\lfloor\log_{2r} \frac{N}{M+k-1}\rfloor\}$ steps.
\end{theorem}

\begin{proof} 
Consider a conceptual set consisting of $N$ processes, each identified by a different original name.
The first goal is to identify a subset~$K$ of these processes consisting of at most~$k$ elements so that an execution of the algorithm by these processes results in a necessary number of local steps. 
The construction of the set~$K$ is recursive and proceeds through a sequence of stages.

A stage represents a group of concurrent reads or writes to shared registers.
A stage determines groups of processes we call \emph{pool} and \emph{residue}.
When the stages get completed, the pool and the residue together make a set~$K$ we seek.
In notation, let \emph{stage $i$} result in determining a \emph{pool set~$P(i+1)$} of processes eligible to be considered for stage $i+1$, a \emph{residue set~$Q(i+1)$}, by determining a prefix~$\cE(i)$ of an execution.
The construction starts with the initial pool~$P(0)$ consisting of $N$ conceptual processes, each identified by its  original name, the initial residue $Q(0)$ is empty, and $\cE(0)$ is  an empty prefix of an execution.

Suppose we have a configuration with a determined sets $P(i)$ and $Q(i)$ and a prefix $\cE(i)$ of the execution.
Let $R(i)$ consist of the processes in $P(i)$ that have a read enabled in the configuration and $W(i)$ consist of the processes in $P(i)$ that have a write enabled.
There are two cases we considered next that determine $P(i+1)$, $Q(i+1)$, and $\cE(i+1)$.

Suppose first that $|R(i)|\ge |W(i)|$.
By the pigeonhole principle, there is a register which a group of at least these many processes want to read:
\[
\frac{|R(i)|}{r}\ge \frac{|P(i)|}{2r}
\ .
\] 
Define the pool $P(i+1)$ to be this group of processes and the residue $Q(i+1)$ to be equal to~$Q(i)$. 
The prefix $\cE(i+1)$ is obtained from $\cE(i)$ by having the processes in $R(i)$ read the register in arbitrary order. 

Suppose next that $|R(i)| < |W(i)|$.
By the pigeonhole principle, there is a register to which  a group of at least these many processes want to write to:
\[
\frac{|W(i)|}{r}\ge \frac{|P(i)|}{2r}
\ .
\] 
Define the pool $P(i+1)$ to be this group of processes.
The prefix $\cE(i+1)$ is obtained from $\cE(i)$ by having the processes in $W(i)$ write to the register in arbitrary order. 
A process~$p$ that writes last in this group is added to $Q(i)$ to obtain $Q(i+1)=Q(i)\cup\{ p \}$.

As the recursive construction progresses, for $i\ge 0$, the following invariant is maintained:
\begin{enumerate}
\item[1)]
the pool $P(i)$ includes at least $\frac{N}{(2r)^i}$ processes, 
\item[2)]
all the processes in~$P(i)$ have exactly the same history in~$\cE(i)$ of reading from  shared registers,
\item[3)] 
the residue $Q(i)$ includes  at most $i$ processes,
\item[4)]
all the processes that wrote a value to a shared register that has ever been read in~$\cE(i)$  are in~$Q(i)$. 
\end{enumerate}

We continue for the maximum number of stages~$t$  such that two constraints are met.
One is that the inequality $\frac{N}{(2r)^{t}}\ge M+k-1$ holds, where $M$ is the range of new names.
The resulting stage number $t$ satisfies $t\le \log_{2r} \frac{N}{M+k-1}$ and $P(t)$ contains at least $M+k-1$ elements.
The other constraint is that $t\le k-2$, so that $Q(t)$ has at most $k-2$ elements. 
These two requirements combined determine 
\[
t=\min\left\{k-2,\Bigl\lfloor\log_{2r} \frac{N}{M+k-1}\Bigr\rfloor\right\}
\ .
\]
The definitions of $P(t)$ and $Q(t)$ imply that the processes in~$P(t) \setminus Q(t)$  have not written in~$\cE(t)$ yet  and have read the same values from the same shared registers in the same order of reading.
The set $P(t) \setminus Q(t)$ has at least $M+k-1-(k-2)\ge M+1$ elements.

Suppose that a decision on a new name is made by each among the processes in $P(t) \setminus Q(t)$ were possible without any further reads or writes.
The range of new names has $M$ elements.
By the pigeonhole principle, there are two processes $p_1$ and $p_2$ in $P(t) \setminus Q(t)$ that would get assigned  the same name.
This means some process among $p_1$ and $p_2$ performs at least one additional read.
We set $K=Q(t)\cup \{p_1,p_2\}$, which has at most $k$ elements.

There is an execution $\cE$ of the algorithm, with the processes in~$K$ as the only contenders for new names, such that $\cE(t)$ restricted only to the events involving the processes in~$K$ is a prefix of~$\cE$.
The processes that wrote values ever read in~$\cE(t)$ are in~$K$.
The processes $p_1$ and $p_2$ have the same history of reads  in~$\cE(t)$, and each of them read $t$ times without writing even once.
The algorithm is a wait-free solution to Renaming, so both $p_0$ and $p_1$ eventually assign themselves  names in~$\cE$.
It follows that at least one of the processes $p_1$ and $p_2$ eventually performs at least one read  after~$\cE(t)$.
This means that the process makes at least $1+t$ steps in execution~$\cE$.
\end{proof} 

Theorem~\ref{thm:rename-lower-bound} implies that some process executing a renaming algorithm has to perform at least $k-1$ steps in a suitable distributed system, which we show next.


\begin{corollary}
\label{cor:lower-bound-k-steps}

For any $k$-renaming algorithm for processes with new names in some range~$[M]$ and using some $r$ shared read-write registers, there exists a sufficiently large range of original names~$[N]$ and an execution in which some process performs at least $k-1$ steps.
\end{corollary}

\begin{proof}
Let us choose $N$ such that the part $\log_{2r} \frac{N}{M+k-1}$ in the lower bound of Theorem~\ref{thm:rename-lower-bound} is at least $k-2$.
To thus end, it suffices to set $N$ to be at least as large as $(M+k-1)(2r)^{k-2}$.
Then the lower bound of Theorem~\ref{thm:rename-lower-bound} becomes $k-1$.
\end{proof}

Algorithm \textsc{Square-Rename}$(k)$ works for arbitrary range $[N]$ of  original names.
Corollary~\ref{cor:lower-bound-k-steps} implies that this algorithm is asymptotically optimal with respect to local-step performance, which is $\cO(k)$.


\begin{corollary}
\label{cor:lower-bound-k-N-renaming}

For any $(k,N)$-renaming algorithm that has $M = \cO( k)$ as a bound on new names and uses $\cO(k\log\frac{N}{k})$ auxiliary registers, for a range of original names $N=\Omega(n^2)$, there exists an execution in which some process performs $\Omega\bigl(\frac{\log n}{\log k +\log\log n}\bigr)$ steps.
\end{corollary}

\begin{proof}
We reformulate the logarithmic part of the bound in Theorem~\ref{thm:rename-lower-bound} as follows:
\begin{equation}
\label{eqn:rename-lower-bound}
\log_{2r} \frac{N}{M+k-1} = \frac{\lg \frac{N}{M+k-1}}{\lg (2r)} = \frac{\lg N-\lg (M +k-1)}{\lg r + 1}
\ .
\end{equation}
By the assumptions, we have $\Omega(1) + 2\lg n \le \lg N$, $\lg (M+k-1) \le \lg (2M) \le \cO(1) + \lg k$,  and $\lg  r \le  \lg k + \lg\lg n +\cO(1)$.
Substitute these into~\eqref{eqn:rename-lower-bound} to obtain a bound
\[
\frac{\lg N-\lg (M +k-1)}{\lg r + 1} \ge\frac{\lg n +\Omega(1)}{\lg k+ \lg\lg n +\cO(1)}
\ ,
\]
which is $\Omega\bigl(\frac{\log n}{\log k+\log\log n}\bigr)$.
\end{proof}

Corollary~\ref{cor:lower-bound-k-N-renaming} implies that algorithm \textsc{Compact-Rename}$(k,N)$ may miss local-step optimality by a factor that is about $\log ^2 k$.
More precisely, assuming additionally that $k=\Omega(\log n)$, the lower bound of Theorem~\ref{thm:rename-lower-bound} becomes $\Omega\bigl(\frac{\log n}{\log k}\bigr)$. Each process executing algorithm \textsc{Compact-Rename}$(k,N)$  performs  $\cO(\log k (\log N + \log k\log^* \frac{N}{k}))$ steps, by Theorem~\ref{thm:compact-rename}. 
We have the following asymptotic  identity $\log ^2 k \cdot \Omega\bigl(\frac{\log n}{\log k}\bigr) =\Omega(\log k (\log N + \log k\log^* \frac{N}{k}))$, under the additional assumption about $k$ that $\log k\cdot \log^* \frac{N}{k} =\cO(\log n)$.


\Paragraph{A lower bound on first storing.}

The lower bound on the first operations of storing applies to implementations of solutions of Store\&Collect in which a process proposes a value by writing it to a dedicated register, one such a register per process, and collecting is performed by reading all such registers.
The first operation \texttt{Store} that a process performs requires identifying a register to write a proposed value. 
We assume there are $r$ shared registers available.
The processes have names from the range~$[N]$.


\begin{theorem}
\label{thm:store-collect-lower-bound}

Let $k$ processes among a total of $n$ in an asynchronous system using $r$ auxiliary shared read-write registers and with names in the range $[N]$ execute a wait-free algorithm for Store\&Collect in a system with $r$ shared read-write registers. 
Then there exists an execution in which some process performs $1+\min\{k-2,\lfloor\log_{2r} \frac{N}{r+k-1}\rfloor\}$ steps in its first operation \texttt{Store}.
\end{theorem}

\begin{proof} 
A proof  is similar to that of Theorem~\ref{thm:rename-lower-bound}.
We structure a prefix of an execution $\cE$ to reflect stages, with a pool $P(i)$ and residue $Q(i)$ sets of processes after stage~$i$.
An execution continues for the maximum number of stages~$t$  such that two constraints are met.
One is that the inequality $\frac{N}{(2r)^{t}}\ge r+k-1$ holds.
The resulting stage number $t$ satisfies $t\le \log_{2r} \frac{N}{r+k-1}$ and $P(t)$ contains at least $r+k-1$ elements.
The other constraint is that $t\le k-2$, so that $Q(t)$ has at most $k-2$ elements. 
These two requirements combined determine 
\[
t=\min\left\{k-2,\Bigl\lfloor\log_{2r} \frac{N}{r+k-1}\Bigr\rfloor\right\}
\ .
\]
The definitions of $P(t)$ and $Q(t)$ imply that the processes in~$P(t) \setminus Q(t)$  have not written in~$\cE(t)$ yet  and have read the same values from the same shared registers in the same order of reading.
The set $P(t) \setminus Q(t)$ has at least $r+k-1-(k-2)\ge r+1$ elements.

Suppose that selections of registers for storing by each among the processes in $P(t) \setminus Q(t)$ were possible without any further reads or writes.
There are $r$ shared registers.
By the pigeonhole principle, there are two processes $p_1$ and $p_2$ in $P(t) \setminus Q(t)$ that would select the same register.
This means some process among $p_1$ and $p_2$ performs at least one additional read.
We set $K=Q(t)\cup \{p_1,p_2\}$, which has at most $k$ elements.

There is an execution $\cE$ of the algorithm, with the processes in~$K$ as the only competitors to select registers for storing, such that $\cE(t)$ restricted only to the events involving the processes in~$K$ is a prefix of~$\cE$.
The processes that wrote values ever read in~$\cE(t)$ are in~$K$.
The processes $p_1$ and~$p_2$ have the same history of reads  in~$\cE(t)$, and each of them read $t$ times without writing even once.
At least one of the processes $p_1$ and $p_2$ eventually performs at least one read  after~$\cE(t)$.
This means that the process makes at least $1+t$ steps in execution~$\cE$.
\end{proof}

Theorem~\ref{thm:store-collect-lower-bound} implies that for each solution of Store\&Collect some process has to perform at least $k-1$ steps during its first storing, in a suitable distributed system, which we show next.


\begin{corollary}
\label{cor:storing-lower-bound-k-steps}

For any Store\&Collect algorithm for a system with some $r$ shared read-write registers, there exists a sufficiently large range of original names~$[N]$ and an execution in which some process performs at least $k-1$ steps.
\end{corollary}

\begin{proof}
Let us choose $N$ such that the part $\log_{2r} \frac{N}{r+k-1}$ in the lower bound of Theorem~\ref{thm:store-collect-lower-bound} is at least $k-2$.
To thus end, it suffices to set $N$ to be at least as large as $(r+k-1)(2r)^{k-2} $.
Then the lower bound of Theorem~\ref{thm:store-collect-lower-bound} becomes $k-1$.
\end{proof}

The algorithm for Store\&Collect with performance characteristics summarized in Theorem~\ref{thm:collect-only-k-known} works for an arbitrary range $[N]$ of original names.
Corollary~\ref{cor:storing-lower-bound-k-steps} implies that this algorithm is asymptotically optimal with respect to local-step performance of first storing, which is $\cO(k)$.


\begin{corollary}
\label{cor:lower-bound-k-N-storing}

For any algorithm for Store\&Collect that  uses $r=\cO(k\log\frac{N}{k})$ auxiliary registers for a range of original names $N=\Omega(n^3)$, there exists an execution in which some process performs $\Omega\bigl(\frac{\log n}{\log k +\log\log n}\bigr)$ steps in its first storing.
\end{corollary}

\begin{proof}
Reformulate the logarithmic part of the bound in Theorem~\ref{thm:store-collect-lower-bound} as follows:
\begin{equation}
\label{eqn:store-lower-bound}
\log_{2r} \frac{N}{r+k-1} = \frac{\lg \frac{N}{r+k-1}}{\lg (2r)} = \frac{\lg N-\lg (r +k-1)}{\lg r + 1}
\ .
\end{equation}
By the assumptions, we have $\Omega(1) + 3\lg n \le \lg N$  and $\lg  r \le  \lg (r+k-1) \le \lg k + \lg\lg n +\cO(1)$.
Substitute these into~\eqref{eqn:rename-lower-bound} to obtain a bound
\[
\frac{\lg N-\lg (r +k-1)}{\lg r + 1} \ge\frac{\lg n +\Omega(1)}{\lg k+ \lg\lg n +\cO(1)}
\ ,
\]
which is $\Omega\bigl(\frac{\log n}{\log k+\log\log n}\bigr)$.
\end{proof}

Corollary~\ref{cor:lower-bound-k-N-storing} implies that the algorithm for Store\&Collect for the case when both $k$ and $N$ are known, and whose performance characteristics are given in Theorem~\ref{thm:collect-k-N-both-known}, may miss local-step optimality of first storing by a factor that is about $\log ^2 k$.
More precisely, assuming additionally that $k=\Omega(\log n)$, the lower bound of Theorem~\ref{thm:store-collect-lower-bound} becomes $\Omega\bigl(\frac{\log n}{\log k}\bigr)$. 
Each process storing for the first time  performs  $\cO(\log k (\log N + \log k\log^* \frac{N}{k}))$ steps, by Theorem~\ref{thm:collect-k-N-both-known}. 
We have the following asymptotic  identity $\log ^2 k \cdot \Omega\bigl(\frac{\log n}{\log k}\bigr) =\Omega(\log k (\log N + \log k\log^* \frac{N}{k}))$, under the additional assumption about $k$ that $\log k\cdot \log^* \frac{N}{k} =\cO(\log n)$.

\section{Unbounded Selection}

We consider problems concerning processes selecting positive integers continuously such that each selection is exclusive.
A selected positive integer may be considered as an abstract name from an unbounded range.
Such names can be used to identify registers from an infinite pool of registers.
Infinitely repeated selections of names are efficient when they minimize the number of positive  integers that  never get selected to be assigned as names.

Next, we review the functionality of a distributed dynamic data structure that we call a ``repository.''
We will refer to two related concepts of ``storing'' and ``depositing'' a value in a register.
A value written to a register gets stored in it as long as it is not overwritten by a different value.
Depositing a value means  storing it indefinitely in some register.
A repository is a concurrent data structure for depositing values in shared read-write registers.
The underlying assumption is that, for each point in time, each process will eventually generate a value to be deposited in a repository.

We assume that there are infinitely many shared read-write registers, denoted $R[1], R[2], R[3],\ldots$, used to store deposited values.
We say that these registers are \emph{dedicated} to depositing and that $i$ in the \emph{index of register~$R[i]$}.
We assume random access to these shared registers, in that an index~$i$ allows to identify the shared register~$R[i]$ and perform a read or write on it.
Every register~$R[i]$ is initialized to \texttt{null}, which is interpreted as the register being empty. 
An algorithm managing repeated deposits may also use additional auxiliary registers.

A formal definition of depositing and a repository refers to an algorithm implementing this operation and supporting registers.
Consider an execution of an algorithm implementing depositing.
A value $x$ is considered \emph{deposited in a register $R$} at an event in the execution, when the following is satisfied in the system's configuration just after the event:
\begin{enumerate}
\item
The value $x$ is stored in register~$R$. 

\item
The value $x$ will not be overwritten in register $R$  by any different value in the following events of the execution.
\end{enumerate}
A repository is a distributed data structure that provides a repertoire of operations for each participating process.
We describe these operations next.

A process~$p$ may invoke an operation $\texttt{Query}_p$ to obtain a new value to be deposited.
An event $\texttt{Return}_p(v)$ returns a value~$v$ for $\texttt{Query}_p$, where $v\ne \texttt{null}$ is a value to deposit.
An event  $\texttt{Return}_p(\texttt{null})$  indicates that there is no value to deposit yet.

A process~$p$ invokes an operation \texttt{Deposit}$_p(v)$ to deposit a value~$v\ne  \texttt{null}$.
An event \texttt{Ack}$_p(i)$ acknowledges completing \texttt{Deposit}$_p()$, where $R[i]$ is the register in which the value has been deposited.
When a process crashes while working to deposit a value, and depositing this value has not been acknowledged before the crash, then the value may either get deposited or not in a dedicated register.

Following the standard understanding of simulating executions, as presented by Attiya and Welch~\cite{Attiya-Welch-book2004}, we prohibit ``pipelining'' on the operations \texttt{Query} and \texttt{Deposit}.
This means that a process may invoke a new operation only after the previous one, if any, has returned an outcome or has been acknowledged as completed.
When a process~$p$ obtains a return of a query for a new value, then the process  eventually invokes \texttt{Query}$_p$ again.
We assume \emph{fair occurrence  of deposit requests} at processes, which means that each process eventually obtains a new value to deposit, after having deposited the previous value, if any, unless the process crashes.

Let us observe that no algorithm depositing values and resilient to even one crash can guarantee for a specific register that a value gets deposited in the register eventually, since otherwise the value stored there could be used to determine  a decision in a solution of Consensus, which is impossible in asynchronous systems with processes prone to crashes and read-write registers~\cite{Attiya-Welch-book2004,HerlihyKozlovRajsbaum-book,Lynch-book96}.
This means that for any execution of a depositing algorithm, some registers dedicated for depositing may never be used for deposits.
Now the question is how many registers dedicated for depositing will never be used to deposit?
There is a simple solution to the problem of repeated deposits in which a process~$i$ deposits only in consecutive registers with indices congruent to~$i$ modulo~$n$.
The problem with this approach is that if even one process crashes then infinitely many registers dedicated for depositing will never deposit a value.
We want to have a solution to repeated depositing in which the number of dedicated registers not used for deposits is finite in any infinite execution.

A \emph{repository} is a concurrent data structure that allows each process to deposit  values in dedicated registers that satisfied the following two properties, subject to fair occurrence  of deposit requests  in an infinite execution:
\begin{description}
\item[\sf Persistence:]  
Starting from an acknowledgment event  \texttt{Ack}$_p(i)$, for a process~$p$ and register $R[i]$ dedicated for depositing, the value stored in~$R[i]$ by $p$ becomes deposited. 
\item[\sf Utilization:] 
There are finitely many registers dedicated for depositing that never store a deposited value.
\end{description}

We want the quality of an implementation of a repository to be be minimally non-blocking, but preferably wait-free.
These qualities have the following standard meaning:
\begin{description}
\item[\sf Non-blocking:] 
If at least one nonfaulty process wants to deposit a value in a configuration, then eventually a value gets deposited.
\item[\sf Wait-free:] 
If some nonfaulty process wants to deposit a value in a configuration, then eventually this process succeeds in depositing its value.
\end{description}
The Repository problem is to implement a repository in a distributed system of $n$ processes prone to crashes and an unbounded supply of read-write registers initialized to \texttt{null}.


\Paragraph{Implementations of a repository.}

The algorithms we give next assign integers to processes such that an assignment provides exclusivity of a selection of an integer by a process.
We interpret a newly assigned integer~$i$ as an indication that the register~$R[i]$ could be  available for depositing.

We start from the algorithm called \textsc{Selfish-Repository}, which works as follows.
Each process~$p$ maintains a sorted list $L_p$ of $2n-1$ indices $i$ in its local memory, which are interpreted as identifying registers possibly still available for deposits.
The list is initialized to store the first $2n-1$ positive integers arranged in order in consecutive entries.
Pointer $L_p.\texttt{head}$ points to the first item on the list.
For an entry $x$ in the list, $x.\texttt{next}$ is the next item in the list, and $x.\texttt{value}$ is the integer stored in the entry.
Process~$p$ also  uses a variable~$A_p$ interpreted as an index of the next available empty register immediately after the registers whose indices are in~$L_p$.
The variable~$A_p$ is initialized to~$2n$.
The entry at position $i$ in list~$L_p$ of process~$p$ is denoted $L_p[i]$ (or $L[i]$ when $p$ is understood).

Process~$p$ may \emph{update list~$L_p$}, which is performed as follows.
Process~$p$ scans~$L_p$, starting from $L_p.\texttt{head}$.
For each entry $x$ scanned in the list, if $x.\texttt{value}=j$, then $p$ reads~$R[j]$. 
If $R[j]$ is still empty, then the next entry $x.\texttt{next}$ is considered, otherwise $p$ removes entry~$x$ from  list $L_p$, by making the predecessor of $x$ point to its successor $x.\texttt{next}$, and begins scanning registers~$R[i]$ one by one in the order of indices, starting from the index~$i$ stored in~$A_p$.
The scan of array~$R$ continues until an empty~$R[k]$ is found, if any.
Once process~$p$ reads an empty $R[k]$ then $p$ appends a new entry $y$ with $y.\texttt{value}=k$  to~$L_p$, and sets~$A_p$ to~$k+1$.
This completes processing entry~$x$.
Next the entry $x.\texttt{next}$ in~$L_p$ is processed in a similar manner.
Updating list~$L_p$ ends after all the entries of~$L_p$ have been processed.

We will use an atomic-snapshot object~$\cS$. 
It includes read-write registers~$S[p]$, for each process~$p$, for $1\le p\le n$, such that $S[p]$ is writable by~$p$ and readable by all processes.
A view returned to a process that invoked taking a snapshot consists of a vector $V=(V[1],V[2],\ldots,V[n])$, where $V[i]$ stores the value read from~$S[i]$.
Each register~$S[i]$, for $1\le i\le n$, is initialized to \texttt{null}.
The registers~$S[i]$ are supported by other auxiliary registers in~$\cS$ to equip~$S$ with the functionality of an atomic snapshot object, see~\cite{AfekADGMS93}.
Process~$p$ writes an integer in the interval $[1,2n-1]$ to the register~$S[p]$ in~$\cS$.
After taking a snapshot using~$\cS$, process~$p$ assigns itself the \emph{rank} defined as follows: it is the rank of $p$ among the indices~$q$ such that the variable~$V[q]$ stores an entry different from \texttt{null}.
A snapshot determines integers in the interval $[1,2n-1]$ that are \emph{available}: these are the numbers in the interval that do not occur in the snapshot.
For each snapshot with at least one repetition of entries there are always at least $n$ integers available: this is because at most $n-1$ numbers in $[1,2n-1]$ occur in the snapshot.

The need for a deposit occurs when a process~$p$ that has been querying for a new value to deposit obtains such a value.
Then process~$p$ begins attempts to acquire an index of a register available for deposits.
Such attempts are performed repeatedly until identifying an index of a register~$i$ exclusively available to~$p$, we call such $i$ a \emph{list name}.
Each attempt begins with $p$ reviewing its list~$L_p$ and setting the \emph{candidate} index, of an eligible register dedicated to depositing, to~$L_p[1]$.
The first goal is to go through a sequence of candidates to eventually identify a list name, while verifying each candidate until it passes the verification and becomes a list name.

Identifying a candidate and verification proceed as follows. 
After $p$ sets the candidate to~$L[1]$, it  repeats the following three actions.
First $p$ sets $S[p]$ to the candidate index.
Then $p$ invokes obtaining a view~$V$ from the snapshot object~$\cS$.
The  third action depends on whether the candidate is unique in the view~$V$.
If this is the case then $p$ treats the candidate index in $L$ as a list name produced by the list.
Otherwise, $p$ sets $r$ to the rank of $p$ in $V$, then sets $j$ to the $r$th integer available in $[1..2n-1]$, and finally chooses a next candidate as the entry on $L$ at position~$j$.
Now $p$ resumes verifying the candidate.
When process~$p$ acquires a list name~$j$ then it checks to see if $R[j]$ is empty.
If this is the case then $p$ deposits at~$R[j]$.
Otherwise, when $R[j]$ already stores a deposited value then $p$ resumes attempts to acquire an index of an eligible register. 

A pseudocode of algorithm \textsc{Selfish-Deposits} is presented in Figure~\ref{fig:alg-selfish-deposit}.
In the pseudocode, the names of variables \texttt{candidate} and \texttt{list-name} half self-explanatory meaning.
The pseudocode is structured as a repeat loop~(1.), which terminates by return of acknowledgement in the last line~(\ref{sel-dep:ack}).


\begin{figure}[t]
\hrule

\FF

\textsf{algorithm} \textsc{Selfish-Repository}

\FF

\hrule

\FF

\begin{enumerate}[nosep]

\item
\texttt{repeat}

\begin{enumerate}[nosep]
\item
\label{sel-dep-update}
update list $L$; 
$\texttt{list-name}\leftarrow \texttt{null}$;
$\texttt{candidate}\leftarrow L[1]$; 
\item
\texttt{repeat}
\begin{enumerate}[nosep]
\item
$S[p]\leftarrow \texttt{candidate}$
\item
obtain a view $V$ from snapshot object $\cS$
\item
\texttt{if} \texttt{candidate} is unique in the view $V$ \texttt{then} $\texttt{list-name}\leftarrow \texttt{candidate}$ \texttt{else} 
\begin{enumerate}[nosep]
\item
$r\leftarrow$ the rank of $p$ in $V$ 
\item
$j\leftarrow$  the $r$th integer available in $[1..2n-1]$
\item
$\texttt{candidate}\leftarrow L[j]$
\end{enumerate}
\end{enumerate}
\texttt{until} $\texttt{list-name} \ne \texttt{null}$
\label{sel-dep-verify-null}
\item
\texttt{if} $R[\texttt{list-name}]=\texttt{null}$  \texttt{then}
\begin{enumerate}[nosep]
\item 
\label{sel-dep-store}
$R[\texttt{list-name}]\leftarrow v$
\item 
$S[p]\leftarrow\texttt{null}$

\item 
\label{sel-dep:ack}
\texttt{return} \texttt{Ack}(\texttt{list-name})
\end{enumerate}
\end{enumerate}
\end{enumerate}
\FF

\hrule

\FF

\caption{\label{fig:alg-selfish-deposit}
Pseudocode of procedure \texttt{Deposit}$(v)$, for a process~$p$, implementing a selfish repository.}
\end{figure}


\begin{theorem}
\label{thm:selfish-deposit}

Depositing based on algorithm \textsc{Selfish-Deposit} is a non-blocking implementation of a repository such that in each execution at most $2n-2$ dedicated deposit registers are not  used for depositing.
\end{theorem}

\begin{proof} 
When a value $v$ gets stored in a register~$R[\texttt{list-name}]$ by a process~$p$ in instruction~(\ref{sel-dep-store}) in Figure~\ref{fig:alg-selfish-deposit}, then this occurs after \texttt{list-name} has been verified to be unique in the view returned by the snapshot object. 
This entry written by $p$ in~$S[p]$ stayed there until after $p$ completed writing to $R[\texttt{list-name}]$.
This means that at most one process could attempt to store a value in this register $R[\texttt{list-name}]$.
The first such process would succeed, as all subsequent attempts, if any, would be prevented by the verification in line~(\ref{sel-dep-verify-null}) of the pseudocode.
This provides the property of Persistence defining a repository.

Next we argue that there are infinitely many successful deposits, assuming fair occurrence  of deposit requests. 
Suppose there is an event after which no deposits occur.
Eventually every process either has crashed or it has a value to deposit, by the assumed fair occurrence of deposit requests.
Each failure to deposit starts a new iteration of the main repeat loop in Figure~\ref{fig:alg-selfish-deposit}, which begins with updating list~$L$ by instruction~(\ref{sel-dep-update}).
As all the non-faulty processes keep updating the lists, while no deposits occur, then eventually all their lists become equal and store the indices of the smallest empty deposit registers.
The values on this list make a set of $2n-1$ natural numbers. 
Let us take the first event when this occurs.
Consider the first following event when each non-faulty process~$p$ has written to~$S[p]$.
Starting from this point in the execution, the ranks of all the processes become fixed in the snapshot object.
Consider the subsequent writes to $S[p]$ by a process~$p$.
If such a write does not produce a unique number in the view, then each next write of a proposed list name is after choosing by rank.
Once the ranks become fixed, each choosing by rank produces a unique entry in the common list~$L$ shared by all the processes.
It follows that eventually some non-faulty process~$p$ acquires a list name.
Now this list name identifies a unique entry in the list $L$ shared by all the processes.
The register dedicated for deposits with the index in this unique entry of $L$ is empty and no other process attempts to use this register, so $p$ completes a deposit.
This contradicts the assumption that there are no deposits after some event.

If a process crashes in the course of depositing, then the process may have identified  a register~$R[k]$ for depositing by acquiring a list name~$k$, but the crash occurred before the value got stored in the register~$R[k]$.
There may be up to $n-1$ such indices~$k$ and the corresponding registers~$R[k]$.
A crashed process~$p$ may have set $S[p]$ in the snapshot object to its candidate, which now will store its value~$S[p]$ throughout the execution.
If such registers $S[p]$ make an initial segment $S[1..k]$, then eventually the first $k$ numbers in lists~$L$ stabilize and each of them is a list name assigned to a crashed process. 
If the next $k$ entries in the lists~$L$ stabilize as well, then these will forever stay as the first $k$ available indices never to be assigned as list names.
Since there are at most $n-1$ crashes and $k\le n-1$,  we have that up to $n-1$ indices of registers dedicated for deposits may never be used as list names.
This means that up to $2n-2$ registers dedicated for depositing may never store a deposited value.
\end{proof}


\begin{theorem}
\label{thm:non-blocking-repository}

For each non-blocking algorithm implementing a repository there exists an execution in which $n-1$ registers dedicated for depositing remain unused. 
\end{theorem}

\begin{proof}
We argue that in each implementation of a repository, at least $n-1$ dedicated registers may be never used for deposits in some execution.
Namely, when a process~$p$ is to deposit by writing to a register~$R[i]$, and a write event to store the value is enabled, we may ``freeze'' the write.
At this point, no  other process~$q$ will want to deposit to $R[i]$, because otherwise after \texttt{ack}$_q(i)$ happens, the pending write of $p$ to store at $R[i]$ might occur as well, which results in overwriting~$R[i]$, in contradiction of the definition of a repository.
This means that if $p$ crashed rather than merely ``freezed,'' then the register $R$ is never used for depositing by any other process. 
Up to $n-1$ crashes can happen, so at least these many registers might never be used for deposits.
\end{proof}

By Theorem~\ref{thm:selfish-deposit}, algorithm \textsc{Selfish-Deposit} leaves  $\cO(n)$ registers dedicated for deposits that remain unused.
This combined with Theorem~\ref{thm:non-blocking-repository} shows that algorithm \textsc{Selfish-Deposit} leaves an asymptotically optimum number of shared registers in the worst case in a perpetual state of not storing a deposited value.

Next, we consider a wait-free implementation of a repository.
We call the algorithm that provides the implementation \textsc{Altruistic-Deposit}.
The algorithm uses some of the mechanisms in  \textsc{Selfish-Deposit} and extends them.
The difference between the two algorithms is what a process~$p$ does with acquired list names.
In algorithm \textsc{Selfish-Deposit}, a process acquiring list names uses it selfishly as address of registers to deposit.
In executions of algorithm \textsc{Altruistic-Deposit}, processes share acquired list names  with other processes to help in their deposits.

Algorithm \textsc{Altruistic-Deposit} consists of two threads. 
One auxiliary thread produces register indices, and the other thread  deposits values.
The two threads are interleaved in a fair manner, in that each process alternates invoking  instructions from the two threads.
Both threads work on an $n\times n$ array \texttt{Help}$[i,j]$, for $1\le i,j\le n$, of shared read-write registers.
The processwa use a snapshot object~$\cS$ to obtain new list names, similarly as in algorithm \textsc{Selfish-Deposit}.
A process~$i$ writes  a verified list name  into \texttt{Help}$[i,j]$  to be used by process~$j$ for its deposits.


\begin{figure}[t]
\hrule

\FF

\textsf{algorithm} \textsc{Altruistic-Repository} : producing thread

\FF

\hrule

\FF

\begin{itemize}[nosep]
\item[]
$\texttt{column}\leftarrow p$
\item[]
\texttt{repeat}
\begin{enumerate}[nosep]
\item
\texttt{if} $\texttt{Help}[p,\texttt{column}]\ne \texttt{null}$ \texttt{then} increment \texttt{column} in a round robin manner \texttt{else}
\begin{enumerate}
\item
update list $L$; 
$\texttt{list-name}\leftarrow \texttt{null}$;
$\texttt{candidate}\leftarrow L[1]$; 
\item
\texttt{repeat}
\begin{enumerate}[nosep]
\item
$S[p]\leftarrow \texttt{candidate}$
\item
obtain a view $V$ from snapshot object $\cS$
\item
\texttt{if} \texttt{candidate} is unique in the view $V$ \texttt{then} $\texttt{list-name}\leftarrow \texttt{candidate}$ \texttt{else} 
\begin{enumerate}[nosep]
\item
$r\leftarrow$ the rank of $p$ in $V$ 
\item
$j\leftarrow$  the $r$th integer available in $[1..2n-1]$
\item
$\texttt{candidate}\leftarrow L[j]$
\end{enumerate}
\end{enumerate}
\texttt{until} $\texttt{list-name} \ne \texttt{null}$

\end{enumerate}
\item
\label{altruistic-verify-availability}
\texttt{if} $R[\texttt{list-name}]=\texttt{null}$ \texttt{then} 
\begin{enumerate}[nosep]
\item
$R[\texttt{list-name}]\leftarrow \texttt{reserved}$
\item
$\texttt{Help}[p,\texttt{column}]\leftarrow\texttt{list-name}$
\end{enumerate}
 \item 
$S[p]\leftarrow\texttt{null}$
\end{enumerate}
\end{itemize}
\FF

\hrule

\FF

\caption{\label{fig:alg-altruistic-repository-produce}
Pseudocode of a producing thread in the altruistic repository, used by a process~$p$ to obtain an index of a register available for depositing.
Such an index gets stored as entry $\texttt{Help}[p,\texttt{column}]$ of the column $\texttt{Help}[p,*]$ of the array \texttt{Help}.}
\end{figure}

A process~$p$ in the producing thread keeps reading the registers in row \texttt{Help}$[p,*]$ in a round-robin manner starting from the diagonal entry. 
If $p$ finds some register \texttt{Help}$[p,c]$ equal to \texttt{null} then $p$ works to obtain a new list name.
After successfully acquiring a list name~$i$, process~$p$ verifies if $R[i]$ is empty, which means it has not been reserved yet.
If this is the case then process~$p$ first marks $R[i]$ as \texttt{reserved} and then writes~$i$ into~\texttt{Help}$[p,c]$.
The value \texttt{reserved} is assumed to be different from \texttt{null}.
A pseudocode of the producing thread is in Figure~\ref{fig:alg-altruistic-repository-produce}.


\begin{lemma}
\label{lem:column-nonempty}

For each process~$p$ and an event it is involved, eventually some entry in the column $\texttt{Help}[*,p]$ stores an index of a reserved register available for deposits. 
\end{lemma}

\begin{proof}
Processes execute producing threads similarly as depositing selfishly in that this produces new reserved registers in a non-blocking manner.
As new entries in the array $\texttt{Help}[*,*]$ get \texttt{reserved} written in them, the writers wrap around their rows in a round robin manner.
The perpetual existence of a column of the array $\texttt{Help}[*,p]$ with all entries \texttt{null} would contradict the non-blocking progress achieved in the execution.
If a process~$q$ obtains a new \texttt{list-name} in an execution of producing thread, then this number is unique in the view provided by the snapshot object. 
This means that no write to $R[\texttt{list-name}]$ by some other process is pending when $q$ reads $R[\texttt{list-name}]$ and finds it empty in instruction \eqref{altruistic-verify-availability} in Figure~\ref{fig:alg-altruistic-repository-produce}, and so eligible to write \texttt{reserved}.
\end{proof}

A process~$p$ in the depositing thread keeps reading the column \texttt{Help}$[*,p]$ in a round-robin fashion, starting from the diagonal entry.
Once $p$ finds an index $j\ne \texttt{null}$ stored at \texttt{Help}$[r,p]$, then $p$ deposits in~$R[j]$ and then writes \texttt{null} to erase value~$j$ in \texttt{Help}$[r,p]$.
A pseudocode of the depositing thread is in Figure~\ref{fig:alg-altruistic-repository-deposit}.


\begin{figure}[t]
\hrule

\FF

\textsf{algorithm} \textsc{Altruistic-Repository} : depositing thread

\FF

\hrule

\FF

\begin{enumerate}[nosep]
\item
$\texttt{row}\leftarrow p$
\item
\texttt{while} $\texttt{Help}[\texttt{row},p]=\texttt{null}$ \texttt{do}
increment \texttt{row} in a round robin manner
\item
$\texttt{index}\leftarrow \texttt{Help}[\texttt{row},p]$
\item
$\texttt{Help}[\texttt{row},p]\leftarrow \texttt{null}$
\item 
$R[\texttt{index}]\leftarrow v$
\item
\texttt{Ack}$(\texttt{index})$
\end{enumerate}
\FF

\hrule

\FF

\caption{\label{fig:alg-altruistic-repository-deposit}
Pseudocode of a depositing thread, used by a process~$p$ to deposit a value~$v$, implementing an altruistic repository.}
\end{figure}


\begin{theorem}
\label{thm:altruistic-deposit}

Depositing based on algorithm \textsc{Altruistic-Deposit} is a wait-free implementation of a repository such that at most $(n+2)(n-1)$ dedicated deposit registers will never be used for depositing.
\end{theorem}

\begin{proof}
Consider an event in which a process~$q$ wants to deposit a value~$v$.
The process invokes the depositing thread and so keeps reading entries in the column $\texttt{Help}[*,q]$ in a round-robin manner.
By Lemma~\ref{lem:column-nonempty}, process~$q$ eventually reads $\texttt{Help}[r,q]=j$,  for some row~$r$ and index~$j$. 
The register~$R[j]$ was verified to be empty by the process~$p$ that wrote~$j$ to $\texttt{Help}[r,q]$, which occured when executing line~\eqref{altruistic-verify-availability} in Figure~\ref{fig:alg-altruistic-repository-produce}.
Process~$q$ can safely store the value~$v$ in~$R[j]$, because while the entry $\texttt{Help}[r,q]$ stays equal to index~$j$, register~$R[j]$ stays equal to \texttt{reserved}, which is different from \texttt{null}.
This prevents the index~$j$ to be written at other locations of the array \texttt{Help}, by instruction~\eqref{altruistic-verify-availability} in Figure~\ref{fig:alg-altruistic-repository-produce}, and so prevents multiple values to be possibly stored in succession at~$R[j]$.

Next, we estimate the number of registers dedicated for depositing that may never be used to deposit a value.
This number if maximized when $n-1$ processes crash while many entries of  \texttt{Help}$[*,*]$ store indices of  registers reserved for depositing.
Let $p$ be the only process that never crashes.
The worst case scenario occurs when each of the crashed processes~$q$ has a full column of~$n$ indices reserved and it crashes when working to produce an index of a register to be placed in column $\texttt{Help}[*,p]$.
Such a process~$q$ may have written some list name~$i$ to $S[q]$ and verified that $R[i]=\texttt{null}$ but crashed before setting $R[i]$ to \texttt{reserved} and so also did not reset $S[q]$ to~\texttt{null}.
Suppose all the lists $L$ stabilized to the same sequence of entries, and further that such indices $i$ make the first $n-1$ entries in the lists. 
Then the first $n-1$ entries will never be removed from the lists along with the first $n-1$ entries available in the lists~$L$.
The registers with theses $2n-2$ indices will stay equal to \texttt{null} forever.
We have obtained $n(n-1)$ reserved registers and $2(n-1)$  empty registers never to be used for depositing.
The total number of registers dedicated for depositing that never store a deposited value could be $n(n-1) + 2(n-1) = (n+2)(n-1)$.
\end{proof}

It is an open problem if there exists a wait-free implementations of a repository that leaves out $o(n^2)$ registers not used for depositing.


\Paragraph{Mining names.}

Next, we consider the task to have processes work continuously to accumulate a possibly  unlimited collection of exclusive names.
The distributed system consists of $n$ processes prone to crashes and a number of shared objects.
We call designing an algorithm for this task the \emph{Mining-Names} problem.
The complete specification is as follows.

A positive integer~$i$ is considered to be \emph{assigned to process~$p$ as a name} when $p$ exclusively commits to integer~$i$ by writing~$i$  in a dedicated local write-once memory variable.
Exclusivity means that no two processes ever commit to the same integer, so a name can be interpreted as an exclusive reservation of a natural number.
Committing to a name resembles committing to a decision in solutions of Consensus.
After committing to a name, a process can proceed to commit to some other natural number as a name as well.
A process participating in acquiring new names never stops voluntarily.

An algorithm is a solution to Mining-Names if in each infinite execution the following two properties are satisfied:
\begin{description}
\item[\sf Naming:] 
No two different processes ever commit to the same integer as a shared name.

\item[\sf Utilization:] 
There are finitely many positive integers that never get acquired as names.
\end{description}
A mining names solution in a system with process crashes could be non-blocking or wait-free, which is understood as follows.
\begin{description}
\item[\sf Non-blocking:] 
For each  event in an execution, eventually a new name gets acquired.
\item[\sf Wait-free:] 
For each  process and any event in an execution, eventually this process acquires a new name.
\end{description}

Algorithms implementing a repository can be adapted to mining names. 
Namely, let every process keep invoking an operation to deposit a dummy value~$v$.
Rather than deposit $v$ in a register~$R[i]$, for some index~$i$, the process commits to the name~$i$.
After a new name has been acquired, the process invokes depositing a dummy value again.
This transformation from depositing to mining names requires the same distributed system that supports depositing.
In particular, the implementation of repository with properties summarized in Theorems~\ref{thm:selfish-deposit} and~\ref{thm:altruistic-deposit} assumes that an infinite array of shared read-write registers each initialized to \texttt{null} is available.


\begin{theorem}
\label{thm:name-mining-solutions}

The Mining-Names problem can be solved by $n$ processes in a non-blocking fashion by an algorithm that leaves at most $2n-2$ nonnegative integers not assigned as names, or in a wait-free manner by an algorithm that leaves at most $(n+2)(n-1)$ integers never assigned as names.
\end{theorem}

\begin{proof}
This follows from combining the general transformation of implementations of a repository to mining names and Theorems~\ref{thm:selfish-deposit} and~\ref{thm:altruistic-deposit}.
\end{proof}

Mining names can be used to implement a repository  by a general transformation, which works as follows.
Every process keeps mining names, and as soon as a new name~$i$ is acquired, this reserves the register~$R[i]$ for depositing.


\begin{theorem}
\label{thm:name-mining-non-blocking-optimality}

For each non-blocking algorithm for Mining-Names, there exists an execution in which $n-1$ positive integers are not assigned as names to any process.
\end{theorem}

\begin{proof}
We apply the general transformation from an algorithm mining names to an implementaiton of a repository.
This transformation creates a non-blocking algorithms mining names from a non-blocking implementation of a repository.
If an algorithm for mining names could guarantee fewer than $n-1$ positive integers never assigned as names then this could be converted into a non-blocking solution for a repository that leaves out at most $n-1$ registers never used for deposits.
This would contradict Theorem~\ref{thm:non-blocking-repository}.
\end{proof}

The implementations of repository we developed have processes use newly acquired list names as indices of registers in an unbounded array of read-write registers.
This works only if a distributed system includes an unbounded array of shared read-write registers, each initialized to \texttt{null}.
The registers dedicated for depositing allow to keep track of indices of used registers, and to obtain next registers still available for deposits by the operation of updating lists.
It is an open problem if there exist algorithmic solutions to Mining-Names in an asynchronous  distributed system with finitely many shared read-write registers and processes prone to crashes.


\bibliographystyle{plain}

\bibliography{asynchronous-selection}

\end{document}